\documentclass[10pt, conference, letterpaper]{IEEEtran}%draftcls
\usepackage{array}

\usepackage{mdwmath}
\usepackage{mdwtab}

\usepackage[dvips]{graphicx}

\usepackage{float}
\usepackage{stfloats}
\usepackage{lipsum}
\usepackage{tabularx}
\usepackage[tight,footnotesize]{subfigure}
\usepackage[small]{caption}%\usepackage[caption=false]{caption}
\usepackage{epstopdf}

\usepackage{amsmath}
\usepackage{amssymb}
\usepackage{paralist}
\usepackage{mathrsfs}
\usepackage{url}
\usepackage{mathbbol}
\usepackage{times}
\usepackage{booktabs}
\usepackage{verbatim}

\usepackage{bm}			% bold math even for greek letters

\usepackage{amsthm}
\newtheorem{theorem}{Theorem}

\newtheorem{definition}{Definition}

\newtheorem{lemma}{Lemma}

\newtheorem{corollary}{Corollary}
\newtheorem{assumption}{Assumption}
\newtheorem{result}{Result}

\usepackage[table]{xcolor}
\usepackage{color,colortbl}

\definecolor{Gray}{gray}{0.9}

\newcommand{\eq}[1]{Eq.~\eqref{#1}}

\newcommand{\bgp}{\texttt{bgp-degree }}
\newcommand{\Dix}{D(i|x)}

 \usepackage[normalem]{ulem}
\begin{document}
%\title{Performance Analysis of Inter-domain SDN}
%\title{Performance Analysis of Inter-domain Routing Centralization}
\title{Can SDN Accelerate BGP Convergence?\\ {\LARGE A Performance Analysis of Inter-domain Routing Centralization}}%{\LARGE A Performance Analysis of Inter-domain Routing Centralization}}% and BGP Convergence Time}
\author{Pavlos Sermpezis\\FORTH, Greece\\sermpezis@ics.forth.gr \and
		Xenofontas Dimitropoulos\\FORTH / University of Crete, Greece\\fontas@ics.forth.gr }

\maketitle
\begin{abstract}
The Internet is composed of Autonomous Systems (ASes) or domains, i.e., networks belonging to different administrative entities. Routing between domains/ASes is realised in a distributed way, over the Border Gateway Protocol (BGP). Despite its global adoption, BGP has several shortcomings, like slow convergence after routing changes, which can cause packet losses and interrupt communication even for several minutes. To accelerate convergence, inter-domain routing centralization approaches, based on Software Defined Networking (SDN), have been recently proposed. Initial studies show that these approaches can significantly improve performance and routing control over BGP. In this paper, we complement existing system-oriented works, by analytically studying the gains of inter-domain SDN. We propose a probabilistic framework to analyse the effects of centralization on the inter-domain routing performance. We derive bounds for the time needed to establish data plane connectivity between ASes after a routing change, as well as predictions for the control-plane convergence time. Our results provide useful insights (e.g., related to the penetration of SDN in the Internet) that can facilitate future research. We discuss applications of our results, and demonstrate the gains through simulations on the Internet AS-topology.

\end{abstract}

\section{Introduction}\label{sec:intro}
The Border Gateway Protocol (BGP) is globally used, since the early days of the Internet, to route  traffic between \textit{Autonomous Systems} (ASes) or \textit{domains}, i.e., networks belonging to different administrative entities. BGP is a distributed, shortest path vector protocol, over which ASes exchange routing information with their neighbors, and establish route paths. 

Although BGP is known to suffer from a number of issues related to security~\cite{Kent-secure-BGP-JSAC-2000,Subramanian-listen-whisper-NSDI-2004}, or slow convergence~\cite{Labozitz-Delayed-convergence-CCR-2000,Kushman-Can-Hear-CCR-2007,Oliveira-Quantifying-Path-Exploration-ToN-2009}, deployment of other protocols or modified versions of BGP is difficult, due to its widespread use, and the entailed political, technical, and economic challenges. Hence, any advances and proposed solutions, should be seamless to BGP.%, in order to have chances to become reality.

Taking this into account, it has been proposed recently that Software Defined Networking (SDN) principles could be applied to improve BGP and inter-domain routing~\cite{Gupta-SDX-CCR-2014,Kotronis-CXP-SOSR-2016,Thai-Decoupling-BGP-Conext-2012,Rothenberg-Revisiting-RCP-HotSDN-2012,Bennesby-Innovating-IDrouting-AINA-2014,Lin-Seamless-Internetworking-Demo-Sigcomm-2013}. The SDN paradigm has been successfully applied in enterprise (i.e., \textit{intra}-AS) networks, like LANs, data centers, or WANs (e.g., Google). However, its application to inter-domain routing (i.e., between different ASes) has to overcome many challenges, like the potential unwillingness of some ASes to participate in the routing centralization. For instance, a small ISP might not have incentives (due to the high investment costs) to change its network configuration. This led previous works on inter-domain SDN to consider (a) partial deployment, only by a fraction of ASes, and (b) interoperability with BGP.

The proposed solutions have demonstrated that bringing SDN to inter-domain routing can indeed improve the convergence performance of BGP~\cite{Kotronis-Routing-Centralization-ComNets-2015}, offer new routing capabilities~\cite{Gupta-SDX-CCR-2014}, or lay the groundwork for new services and markets~\cite{Gibb-Outsourcing-NF-HotSDN-2012,Kotronis-CXP-SOSR-2016}. However, most of previous works are system-oriented: they propose new systems or architectures, and focus on design or implementation aspects. Hence, despite some initial evaluations (e.g., experiments, emulations, simulations) we still lack a clear understanding about the interplay between inter-domain centralization and routing performance.
%Moreover, performance evaluation in these works is usually done through experiments, emulations, or simulations. Although, these methods enable an accurate evaluation, they are time and resource demanding, are not scalable, and are not sufficient to provide generic answers about performance.

To this end, in this paper, we study \textit{in an analytic way} the effects of centralization on the performance of inter-domain routing. We focus on the potential improvements on the (slow) BGP convergence, a long-standing issue that keeps on concerning industry and researchers~\cite{survey-bgp-nanog}. Our goal is to complement previous (system-oriented) works, obtain an analytic understanding, and answer questions such as: \textit{``To what extent can inter-domain centralization accelerate BGP convergence? How many ASes need to cooperate (partial deployment) for a significant performance improvement? Is the participation of certain ASes more crucial? Will all ASes experience equal performance gains?''} Specifically, our contributions are:

%The question is how much? Can it improve always? For whom? Experiments/emulations/simulations give some indications or inittial results, but are not able to give annswers for generic cases. Also, using only sims etc., that are time and resources demaninding, makes a generic evaluation (comprising sensitivity analysis, etc.) a heavy task. There is a scaling problem.

%To this end, we attack the problem in an analytic way, which has not be done before. We build a model and conduct an analysis to study what are the effects of routing centralization on BGP convergence, or the establishment of connectivity after a routing change.   Our goals are to obtain analytical understanding and provide insighits about the effects of the network oarameters, e.g., SDN participation, routing paths, AS topology, on the BGP convergence

%Contributions:
\begin{itemize}
\item We propose a model (Section~\ref{sec:model}) and methodology (Sections~\ref{sec:data-plane} and~\ref{sec:control-plane}) for the performance analysis of inter-domain routing centralization. To our best knowledge, we are the first to employ a probabilistic approach to study the performance of inter-domain SDN. %We deem our approach can be used in the future for analyzing different/further aspects of inter-domain SDN.

\item We analyse the time that the network needs to establish connectivity after a routing change. In particular, we derive upper and lower bounds for the time needed to achieve data-plane connectivity between two ASes (Section~\ref{sec:data-plane}), and exact expressions and approximations for the time till control-plane convergence over the entire network (Section~\ref{sec:control-plane}). Our results are given by closed-form expressions, as a function of network parameters, like network size, path lengths, and number of~SDN~nodes.

\item Based on the theoretical expressions, as well as on extensive simulation results, we provide insights for potential gains of centralization, inter-domain SDN deployment strategies, network economics, etc.

% demonstrate the interplay between network parameteres provide insights on the effects of routing centralization, related to potential performance improvements by SDN, topological characteristics of the SDN ASes, network economics, etc. %We demonstrate the applicability of our predictions to real setting, through an extensive set of simulations.
\end{itemize}

We believe that our study can be useful in a number of directions. Research in inter-domain SDN can be accelerated and facilitated, since a fast performance evaluation with our results can precede and limit the volume of required emulations/simulations. The probabilistic framework we propose can be used as the basis (and be extended and/or modified) to study other problems or aspects relating to inter-domain routing, e.g., BGP prefix hijacking, or anycast. Finally, the provided insights can be taken into account in the design of protocols, systems, architectures, pricing policies, etc.

\section{Model}\label{sec:model}
\subsection{Network}
We consider a network, e.g., the whole Internet or a part of it, that consists of $N$ autonomous systems (ASes). We represent each AS as a \textit{single node} that operates as a BGP router; this abstraction that is common in related literature~\cite{Labozitz-Delayed-convergence-CCR-2000,Kotronis-Routing-Centralization-ComNets-2015}, allows to hide the details of the intra-AS structure and functionality, and focus on inter-domain routing. When two ASes are connected (transit, peering, etc., relation), we consider that a link exists between the corresponding routers, over which data traffic and BGP messages can be exchanged. %In reality, an AS can be a very large network composed of hundreds or thousands switches/routers, and extend over a large area. However, the  abstraction of our model allows to hide the details of the intra-AS structure and functionality, whose effects on inter-domain routing is less important.

%We consider a network, e.g., the whole Internet or a part of it, that consists of $N$ autonomous systems (ASes). Since we are interested in the inter-domain routing, we represent each AS as a \textit{single node} that operates as a BGP router, which is common when studying inter-domain routing~\cite{Labozitz-Delayed-convergence-CCR-2000,Kotronis-Routing-Centralization-ComNets-2015}. When two ASes are connected (transit, peering, etc., relation), we consider that a link exists between the corresponding routers, over which data traffic and BGP messages can be exchanged. In reality, an AS can be a very large network composed of hundreds or thousands switches/routers, and extend over a large area. However, the  abstraction of our model allows to hide the details of the intra-AS structure and functionality, whose effects on inter-domain routing is less important.

\subsection{SDN Cluster}
%\subsection{Inter-domain Routing Centralization}

ASes can be ISPs, enterprises, CDNs, IXPs, etc., belong to different administrative entities, and span a wide range of topological, operational, economic, etc., characteristics. As a result, not all ASes should be expected to be willing to cooperate for and/or participate in an inter-domain centralization effort. Routing centralization is envisioned to begin from a group of a few ASes cooperating with each other, e.g., at an IXP location~\cite{Gupta-SDX-CCR-2014,Kotronis-CXP-SOSR-2016}; then, more ASes could be attracted (performance or economics related incentives) to join the group, or form another group.
%a centralized inter-domain routing service

%ASes in the Internet belong to different administrative entities (e.g., ISPs, enterprises, CDNs, IXPs), spanning a wide range of topological, operational, financial, etc., characteristics. Hence, not all ASes should be expected to be willing to cooperate in order to establish a globally centralized inter-domain routing system. Routing centralization is more probable to begin from a group of a few ASes cooperating with each other, e.g., at an IXP location~\cite{}; then, more ASes could be attracted (due to incentives related to performance or financial factors) to join the group, or form another group.

To this end, we assume that $k\in[1,N]$ ASes, i.e., a fraction of the entire network, cooperate in order to centralize their inter-domain routing. In the remainder, we refer to the set of these $k$ ASes, as the \textit{SDN cluster}\footnote{Although we use the term \textit{SDN}, our framework does not require necessarily that routing centralization is implemented on an SDN architecture.}. To avoid delving into system-specific issues of the centralization implementation, we assume the following setup, which captures main characteristics of several proposed solutions(e.g.,~\cite{Kotronis-Routing-Centralization-ComNets-2015,Rothenberg-Revisiting-RCP-HotSDN-2012,fibbing-sigcomm-2015}), and is generic enough to accommodate future solutions: ASes in the SDN cluster exchange routing information with a central entity, which we call \textit{multi-domain SDN controller}. The multi-domain SDN controller might be an SDN controller that directly controls the BGP routers of the ASes (e.g., as in~\cite{Kotronis-Routing-Centralization-ComNets-2015}), or a central server that only provides information or sends BGP messages to the ASes (e.g., similar to~\cite{fibbing-sigcomm-2015}). %\blue{In the latter case, the ASes in the SDN cluster would not need to grant access to their routers, or disclose all their routing policies.}

%there exists a \textit{multi-domain SDN controller}, which is connected to (and controls) the BGP routers of all the ASes in the SDN cluster. \blue{[say that we do not necessarily mean a centralized SDN controller that controls all the routers of the members; such a service could be implemented as a centralized service (e.g., a server running some routing software) that provides information to the participants. E.g., the central server provides information about the routing changes (in a far part of the Internet) to the intra-domain SDN controller / BGP routers / etc. of its member. Then the member, which is controls its own routers, decides based on this information. In this way, no routing policies need to be disclosed etc.]}

\subsection{BGP Updates}
Each node has a routing table (Routing Information Base, RIB), in which each entry contains an IP Prefix, and the corresponding AS-path (i.e., sequence of ASes) through which this prefix can be reached. RIBs are built from the information received by the neighbor ASes: upon a routing change, the ``source'' AS (e.g., the AS that originates a prefix) sends BGP updates to its neighbors to notify them about the change; when an AS receives a BGP update, it calculates the needed updates (if any) for its RIB, and sends BGP updates to its own neighbors. In this way, BGP updates propagate over the entire network, and paths to prefixes are built in a distributed way.

%Each AS/router keeps a routing table (Routing Information Base, RIB) with information about how to route traffic to different IP prefixes. Each entry in the table includes an IP Prefix, and the corresponding AS-path, i.e., the sequence of ASes through which this prefix can be reached. The routing tables are built based on the information received by the neighbor ASes: upon a routing change, the ``source'' AS (e.g., the AS that announced a prefix) sends BGP messages/updates to its neighbors to notify them about the change; when an AS receives a BGP update, it calculates the needed updates (if any) for its RIB, and sends BGP messages to its own neighbors. In this way the BGP updates propagate over the whole network, and (shortest) paths to prefixes are built in a distributed way.

Let us assume that an AS receives a BGP update at time $t_{1}$ and forwards it to a neighbor AS at time $t_{2}$. We call \textit{BGP update time}, and denote $T_{bgp}$, the time between the reception of a BGP update in an AS and its forwarding to a neighbor AS, i.e., $T_{bgp} = t_{2}-t_{1}$. The BGP update times may vary a lot among different ASes and/or connections, since they depend on a number of parameters: routers' hardware/software (e.g., time to process BGP data and update RIB) and/or configuration (e.g., MRAI timers), intra-AS network structure (e.g., number of routers, topology) and/or operation (e.g,. iBGP configuration, intra-AS SDN), etc. 

%Knowing all these parameters and for every AS in the network is not possible. Even if we knew them, the involved complexity would not allow us to perform a tractable analysis for the propagation of the routing information. To this end, we follow a probabilistic approach to model the BGP update times. Specifically, we make the following assumption.

Knowing all these parameters for every AS is not possible, and using (upper) bounds for $T_{bgp}$ would not lead to practical conclusions~\cite{Labozitz-Delayed-convergence-CCR-2000}. Thus, to be able to perform a useful analysis, we follow a probabilistic approach, and model the BGP update times as follows.

\begin{assumption}[BGP updates - renewal process]\label{assumption:t-bgp}
The BGP update times $T_{bgp}$ are independent and identically distributed random variables, drawn from an arbitrary distribution $f_{bgp}(t)$, with $E[T_{bgp}] = \mu_{bgp}$.
%$f_{bgp}:[0,+\infty)\rightarrow [0,1]$, where $\int_{0}^{+\infty} f_{bgp}(x)dx=1$, and $E[T_{bgp}] = \mu_{bgp}$.
\end{assumption}
Under Assumption~\ref{assumption:t-bgp}, BGP update times are given by a renewal process. The model is very generic, since it allows to use any valid function $f_{bgp}(t)$, and thus describe a wide range of scenarios with different parameters. Real measurements can be used to make a realistic selection for the distribution $f_{bgp}(t)$, as we show in Appendix~\ref{sec:distr-t-bgp}; however, a detailed study for fitting the $f_{bgp}(t)$ is beyond the scope of this paper.

%\blue{[say since bgp update times are mostly related to intra-domain characteristics, the independence assumption is not far from reality(???)]}
%\blue{[see comment about autocorrelation - I'd suppose it's not a problem for a single bgp update]}

\subsection{Inter-domain SDN Routing}
Routing information in the SDN cluster propagates in a centralized way, through the multi-domain SDN controller. When an AS in the SDN cluster receives a BGP update from a neighbor AS (not in the SDN cluster), it forwards the update to the SDN controller. The SDN controller, which is aware of the topology in the SDN cluster and the connections/paths to external ASes, informs every AS in the SDN cluster about the needed changes in the routing paths. The ASes that receive the updated routes from the controller, notify their non-SDN neighbors using the standard BGP mechanism. 

%The propagation of the routing information among the ASes in the SDN cluster takes place in a centralized way, through the multi-domain SDN controller. When an AS belonging to the SDN cluster receives a BGP update from an neighbor AS (not in the SDN cluster), forwards the update to the SDN controller. The SDN controller, which is aware of the topology in the SDN cluster and the connections/paths to external ASes, calculates the needed changes in the routing tables and informs every AS in the SDN cluster. The ASes that receive the updated routes from the controller, notify their non-SDN neighbors using the standard BGP mechanism.

Let $t_{1}$ be the time that the first AS belonging to the SDN cluster receives a BGP update from a non-SDN neighbor, and $t_{2}$ the time till \textit{all} ASes in the SDN cluster have been informed (by the controller) for the BGP updates. We denote as $T_{sdn}$ the time needed for all the SDN cluster to be informed after a member has received a BGP update, i.e., $T_{sdn} = t_{2}-t_{1}$. The times $T_{sdn}$ would depend on the system implementation. However, it was shown that system designs can achieve $T_{sdn}\ll T_{bgp}$~\cite{supercharge-me-2015}. Hence, in the remainder -for the sake of presentation- we assume that $T_{sdn}\rightarrow 0$. Nevertheless, our results can be easily modified for arbitrary $T_{sdn}$ (even for cases with $E[T_{sdn}] > E[T_{bgp}]$), without this affecting the main conclusions of the study.

%Thus, we assume -for simplicity- that $T_{sdn}\rightarrow 0$. 

%For instance, in~\cite{Kotronis-Routing-Centralization-ComNets-2015} $T_{sdn}$ is not more than a few seconds, whereas the default value for MRAI timers in Cisco routers is $30sec.$.

%\blue{[Say that (a) T-{sdn} is similarly given by a distribution, (b) in some case it might be much smaller than T-bgp, (c) in the remainder we assume (for the sake of presentation) that T-sdn << T-bgp, (d) however, all our results can be easily modified for arbitrary T-sdn (e.g., even for T-sdn >> T-bgp, (e) maybe cite the TechRep where you mention what are the needed modifications, (f) the main conclusions/insights are not affected by considering Tsdn << T-bgp], (g) remove the argument of MRAI timers}

%This time can be expected to be much lower than the BGP updating process, thus, for simplicity, we assume here that $T_{sdn}=0$.
%This time can be expected to be in the order of few seconds~\cite{}, and much lower than the BGP updating process (cf., the default value for MRAI timers is Cisco routers is $30sec.$), thus, for simplicity, we assume here that $T_{sdn}=0$.

\subsection{Preliminaries and Problem Statement}
\noindent In our analysis, we consider the following setup: 

Every node in the network knows at least one (BGP) path to every other node.

A node initiates a routing change that affects the inter-domain routing (e.g., node $n_{0}$ in Fig.~\ref{fig:sd-path}). This could be an announcement or withdrawal of an IP prefix, an interruption of an AS connection (e.g., a link is down), etc. Here, we consider that a node, which we call the ``source node'', announces a new IP prefix; this routing change affects the entire network, every node will install a path for this prefix in its RIB upon the reception of the BGP update.

Nodes in the SDN cluster, receive route information from the SDN controller, and add an entry in their RIB for the prefix to the source node; even if the path is not established in the node preceding in this path (e.g., in Fig.~\ref{fig:sd-path} node $n_{j}$ might receive the update before node $n_{j-1}$). In this case only the node in the SDN cluster knows how to route traffic to the new prefix, therefore, if the SDN node sends traffic to the new prefix, this would not necessarily reach the source-node. The connectivity will be established when every AS in the path has been informed about the BGP update.

BGP updates do not propagate backwards in the path; this would create loops or longer paths, which are discarded or not preferred by BGP.

We call ``SD-path'' the final path, i.e., the shortest conforming to the routing policies, between the source node (``S'') and another node (``destination'', or ``D'').

In the remainder of the paper we investigate the effects of routing centralization on: (a) the data-plane connectivity between the source node (``S'') and any node (``D'') in the network, i.e., the time needed till all nodes in an SD-path have installed the updated BGP paths after a routing change (Section~\ref{sec:data-plane}); and (b) the control-plane convergence, i.e, the time needed till the entire network has established the final paths corresponding to the routing change (Section~\ref{sec:control-plane}).
 
%\begin{itemize}
%\item every AS in the network knows at least a path to reach every other AS
%\item an AS announces a new prefix
%\item an AS in the SDN cluster when it receives the BGP update, adds an entry in its RIB for the prefix to the source-AS; even if the path is not established in the ASes preceding in this path (in this case the SDN AS knows how to route traffic to the new prefix, although other ASes do not. Hence if the SDN AS sends traffic this would not necessarily reach the source-AS. The connectivity will be established when every AS in the path has been informed about the BGP update.
%\item no AS sends BGP updates to its neighbors backwards in the path, since this would create loops or longer paths (which are discarded or not preferred by BGP)
%\item an SD-path is the final path (i.e., the shortest conforming to the routing policies) between the nodes S and D
%\end{itemize}

For ease of reference, we summarize the notation in Table~\ref{table:important-notation}.
%For ease of reference, we summarize the main notation we use in the paper in Table~\ref{table:important-notation}.
\begin{table}
\centering
\caption{Important Notation}
\label{table:important-notation}
\begin{tabular}{|l|l|l|}
\hline
{$N$}	& {network size (total \# of nodes)}&{}\\
\hline
{$k$}	& {SDN cluster size (total \# of SDN-nodes)}&{}\\
\hline
{$T_{bgp}$}	& {BGP update time}&{}\\
\hline
{$f_{bgp}(t)$}	& {distribution of BGP update times}&{Assumption~\ref{assumption:t-bgp}}\\
\hline
{$d$}	& {path length}&{}\\
\hline
{$k^{'}$}	& {\# of SDN-nodes \textit{on a path}}&{}\\
\hline
{$T_{SD}$}	& {data-plane connectivity time in a SD-path}&{Theorem~\ref{thm:sd-path-d-k}}\\
\hline
{$T_{c}$}	& {BGP convergence time}&{Theorem~\ref{thm:expectation-and-variance-Tc}}\\
\hline
{$T_{\ell}$}	& {$\ell$-partial BGP convergence time}&{Corollary~\ref{thm:l-partial-Tc}}\\
\hline
\end{tabular}
\end{table}

\section{Data-Plane Connectivity}\label{sec:data-plane}
\subsection{Analysis}

A source node ``S'' announces a new IP prefix, and SD-path is the final path from S to a ``destination'' node D; see, e.g., Fig.~\ref{fig:sd-path}. Theorem~\ref{thm:sd-path-d-k} bounds the expectation of the time $T_{SD}$ needed to establish data-plane connectivity in the path.

\begin{figure}
\includegraphics[width=\linewidth]{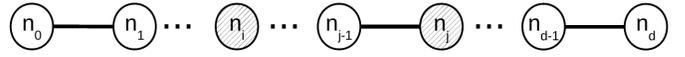}
\caption{\textit{SD path} of size $d$. The node $n_{0}$ initiates the routing change; nodes $n_{i}$ and $n_{j}$ belong to the SDN cluster.}
\label{fig:sd-path}
\end{figure}

%
%\begin{theorem}
%The expectation of the time $T_{SD}$ needed for data plane connectivity in a path from the source node S, which initiates the routing change, to a destination node D, which is at distance $d$ AS-hops from S, is bounded as follows
%\begin{equation}
%\frac{d}{1+\frac{(d+1)\cdot k}{N}}  \leq \frac{E[T_{SD}|d]}{E[T_{bgp}]} \leq (d+1)\cdot \left(1-\frac{k}{N}\right)
%\end{equation}
%\end{theorem}

\begin{theorem}\label{thm:sd-path-d-k}
The expectation of the time $T_{SD}$ in a path of length $d$ with $k^{'} \in [0,d+1]$ nodes in the SDN cluster, is bounded as follows
\begin{equation}
\hspace{-0.08cm}LB(d,k^{'})\cdot E[T_{bgp}] \leq E[T_{SD}|d,k^{'}] \leq UB(d,k^{'})\cdot E[T_{bgp}]
\end{equation}
where 
\begin{equation}
LB(d,k^{'})= \left\{
\begin{tabular}{ll}
$0$				&, $d\leq k^{'} \leq d+1$\\
$\frac{d}{k^{'}+1}$	&, $0\leq k^{'}< d$
\end{tabular}
\right.
\end{equation}
and
\begin{equation}
UB(d,k^{'})= \left\{
\begin{tabular}{ll}
$d-k^{'}+1$	&, $2\leq k^{'}\leq d+1$\\
$d$				&, $0\leq k^{'} <2$
\end{tabular}
\right.
\end{equation}
\end{theorem}
\begin{proof}
The proof is given in Appendix~\ref{sec:proof-of-thm-sd-path-d-k}.
\end{proof}

We provide the intuition behind the proof of Lemma~\ref{thm:sd-path-d-k} in relation to Fig.~\ref{fig:sd-path}. When a node in the SD-path that belongs to the SDN cluster receives the BGP update (e.g., node $n_{i}$), then every other node in the SDN cluster (e.g., node $n_{j}$) is informed about the update, sometimes even before its preceding node(s) (e.g., $n_{j-1}$). Hence, the BGP update can propagate on \textit{different sections} of the SD-path \textit{simultaneously} (e.g,. from $n_{i}$ up to $n_{j-1}$, and -at the same time- from $n_{j}$ to $n_{d}$). The length of these SD-path sections (which determine the BGP update propagation time) depend on the positions of the SDN nodes on the path. The bounds are derived based on the ``best'' and ``worst'' possible positions of the SDN nodes on the SD-path.

%Theorem~\ref{thm:sd-path-d-k} provides bounds for the expected time needed to establish data plane connectivity between two ASes, after one of them has initiated a routing change. Given $d$ and $k^{'}$, one can calculate the range within which the data plane connectivity time is expected to be.

% the length $d$ of the path over which the BGP updates propagate (i.e., the shortest path conforming to routing policies) between the two ASes, and the number of ASes on the path that belong to the SDN cluster, one can calculate the range within which the data plane connectivity time is expected to be.

%\vspace{\baselineskip}
\subsection{Network Topology and Routing Centralization}\label{sec:topo-and-sdn}
Based on Theorem~\ref{thm:sd-path-d-k} we can calculate the average time $E[T_{SD}|d]$ over all paths of the same size $d$ (or, equivalently, for an average path of size $d$), using the property of conditional expectation:%\footnote{The same quantity, gives also the expectation of $T_{SD}$, when the exact set of nodes belonging to the SDN cluster (and, thus, $k^{'}$) is not known.}
\begin{equation}\label{eq:Tsd-conditional-k}
E[T_{SD}|d] = \sum_{i=0}^{d+1} E[T_{SD}|d,k^{'}=i] \cdot P\{k^{'}=i|d\}
\end{equation}
where $P\{k^{'}=i|d\}$ denotes the probability that $i$ nodes (out of the total $d+1$ nodes on the path) belong to the SDN cluster.

\textbf{Topology-independent SDN cluster.} If the SDN cluster is formed independently of the network topology, the quantity $k^{'}$ follows an \textit{hypergeometric distribution} with parameters $N$ (population size), $k$ (number of successes in the population), and $d+1$ (number of draws), and probability mass function
\begin{equation}\label{eq:hypergeometric-distribution}
P\{k^{'}=i|d\} = \frac{\binom{k}{i}\cdot \binom{N-k}{d+1-i}}{\binom{N}{d+1}}
\end{equation}
%and mean value
%\begin{equation}
%E[k^{'}] = \frac{(d+1)\cdot k}{N}
%\end{equation}

\textbf{Topology-related SDN cluster.} On the other hand, if the participation of ASes in the SDN cluster is related to the topology, e.g., because ASes are explicitly selected based on topological characteristics (e.g.,  centrality), or the incentives of cooperation are inherently related to their connectivity (e.g., SDN deployment on tier-1 ISPs, or IXPs~\cite{Gupta-SDX-CCR-2014,Kotronis-CXP-SOSR-2016}), then $k^{'}$ might not be captured accurately by \eq{eq:hypergeometric-distribution}. Therefore, the actual distribution $P\{d,k^{'}\}$ needs to be calculated; however, this might be a difficult (or infeasible) task. 

Alternatively, in certain cases, the distribution $P\{k^{'}=i|d\}$ could be approximated with variations of the standard hypergeometric distribution that are able to take into account the fact that different nodes appear in shortest paths with different probabilities. For instance the \textit{Fisher's noncentral hypergeometric distribution} can be used to consider biased selection of ASes for the SDN cluster: let $\omega_{i}$ be the betweenness centrality~\cite{Newman:Networks-book} of a node $n_{i}$, and $\omega_{sdn}$ and $\omega_{bgp}$ the averages among the nodes in the respective sets, i.e., 
\begin{equation*}
\omega_{sdn} = \frac{\sum_{n_{i}\in SDN} \omega_{i}}{|\{n_{i}:n_{i}\in SDN\}|}~,~~\omega_{bgp} = \frac{\sum_{n_{i}\notin SDN} \omega_{i}}{|\{n_{i}:n_{i}\notin SDN\}|}
\end{equation*}
Denoting $\omega = \frac{\omega_{sdn}}{\omega_{bgp}}$, the probability $P\{k^{'}=i\}$ is approximately given by
\begin{equation}\label{eq:fisher-distribution}
P\{k^{'}=i|d\} = \frac{\binom{k}{i}\cdot \binom{N-k}{d+1-i}\cdot \omega^{i}}{\sum_{j = 0}^{d+1}  \binom{k}{j}\cdot \binom{N-k}{d+1-j}\cdot \omega^{j}}
\end{equation}

In the above distribution, the higher the betweenness centrality of the ASes in the SDN cluster, the more skewed towards the higher values of $k^{'}$ the distribution $P\{k^{'}|d\}$ is, and, thus, the lower the delay $T_{SD}$.

\textbf{Internet AS-topology vs. SDN cluster.} We now focus on the Internet topology, which is of higher interest, and apply our -generic- theoretical results to investigate the effects of routing centralization.

We first build the Internet AS graph from a large experimentally collected dataset~\cite{AS-relationships-dataset} (consisting of $N=55567$ ASes), and infer routing policies over existing links based on the Gao-Rexford conditions~\cite{stable-internet-routing-TON-2001} (this is the most common approach in related literature; more details can be found in Appendix~\ref{sec:simulations-internet}). We consider about $10^{6}$ different SD-paths, from which we calculate the path length distribution $P\{d\}$ (see Fig.~\ref{fig:path-length-distribution}), and the betweenness centrality for each node.

We consider different scenarios with variable SDN cluster size $k = 1, ..., N$, where the set of nodes in the SDN cluster are selected (a) \textit{randomly}, or (b) based on their \textit{betweenness centrality} (i.e., the top $k$ nodes with the highest betweenness centrality values). From Theorem~\ref{thm:sd-path-d-k}, we calculate the lower and upper bounds for the average $T_{SD}$ time over all path lengths, i.e., $E[T_{SD}] = \sum_{d}E[T_{SD}|d]\cdot P\{d\}$, where $E[T_{SD}|d]$ is given by \eq{eq:Tsd-conditional-k}, and $P\{k^{'}|d\}$ from \eq{eq:hypergeometric-distribution} or \eq{eq:fisher-distribution} for the aforementioned cases (a) and (b), respectively.

In Fig.~\ref{fig:bounds-random-betweenness-CAIDA}, we present the lower (LB) and upper (UB) bounds for $E[T_{SD}]$ for different SDN cluster sizes $k$, normalized over the case without routing centralization ($k=0$). When the nodes in the SDN cluster are selected \textit{randomly}, i.e., independently of the topology, a significant decrease in the average connectivity time can be achieved only when at least $20\%$ (around $k=10000$) of the nodes participate in the SDN cluster (note the log scale of the x-axis). This observation, which is in accordance with previous findings~\cite{Kotronis-Routing-Centralization-ComNets-2015}, is a rather grim message for the efficiency of (a randomly deployed) inter-domain routing centralization, since even if a few hundreds or thousands of ASes were willing to cooperate, the gains would be marginal.

On the contrary, as is shown in Fig.~\ref{fig:bounds-random-betweenness-CAIDA}, when the SDN cluster consists of ASes with high \textit{betweenness centrality}, with only a few tens of nodes the average delay can decrease up to $50\%$. This new insight (compared to previous understanding of the effects of routing centralization) brings optimism for the feasibility of inter-domain centralization: even if deployed incrementally, e.g., starting from a few tier-1 ISPs\footnote{Large ISPs are central in the Internet topology, with high betweenness centrality. For example, the top-10 ASes with the highest betweenness centrality values belong to the list of the top-50 ASes with the largest number of ASes in customer cone~\cite{as-rank-website}}, the Internet can immediately see significant performance improvements.

\begin{figure}
\centering\includegraphics[width=\linewidth]{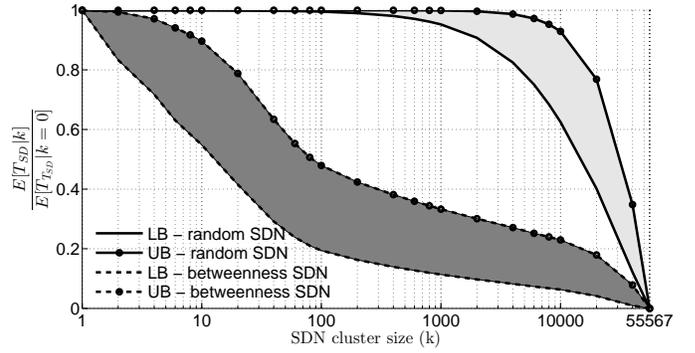}
\caption{Bounds for the average data-plane connectivity time, normalized over the no SDN scenario, i.e., $\frac{E[T_{SD}|k]}{E[T_{SD}|k=0]}$, in the Internet AS-graph. Upper (UB) and lower (LB) bounds enclose the colored areas: nodes in the SDN cluster are selected (i) \textit{randomly} (light grey area) and (ii) with decreasing \textit{betweenness centrality} (dark grey area).}
\label{fig:bounds-random-betweenness-CAIDA}
\end{figure}

\begin{figure}
\centering
\includegraphics[width=\linewidth]{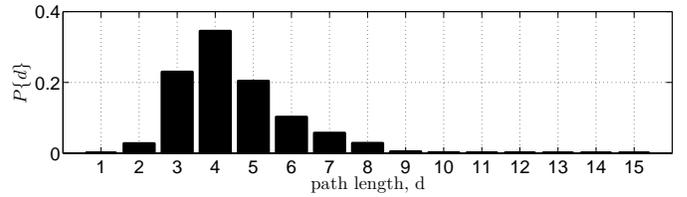}
\caption{Path length distribution on the Internet AS-topology.}
\label{fig:path-length-distribution}
\end{figure}

\subsection{Simulation Results and Implications}
To validate our theoretical results, we conduct simulations on scenarios with varying (a) \textit{network topologies}: synthetic graphs such as full-mesh, Poisson graph, Barabasi-Albert (power low graph), Newman-Watts-Strogatz (small world graph), as well as, the real Internet AS-graph; (b) \textit{SDN cluster sizes}: $k=0,...,N$; and (c) \textit{distributions} $f_{bgp}(t)$: exponential with rate $\lambda=1$ and uniform in $[0,2]$, both with $\mu_{bgp}=1$. In the following we present a subset of representative results, and discuss some important observations.
%(a) network topologies: full-mesh, Poisson graph, regular, star, Barabasi-Albert (power low graph), Newman-Watts-Strogatz (small world graph); 

The average values of $T_{SD}$ in the simulations, are \textit{always} within the bounds of Theorem~\ref{thm:sd-path-d-k} for all pairs $\{d,k^{'}\}$ in every scenario we tested.

In Fig.~\ref{fig:bounds} we compare the simulation results for $E[T_{SD}|d]$ (average over all $k^{'}$) against the theoretical bounds, which are calculated from \eq{eq:Tsd-conditional-k} by using the expressions of \eq{eq:hypergeometric-distribution} (topology-independent SDN cluster) and Theorem~\ref{thm:sd-path-d-k}. For both cases of $f_{bgp}(t)$ %(exponential in Fig.~\ref{fig:bounds-exp} and uniform in Fig.~\ref{fig:bounds-uni})
, the bounds are very tight for $k=50$, when only a small fraction ($5\%$) of the nodes belong to the SDN cluster (top plots). For larger SDN cluster sizes ($k=200$, or $20\%$; bottom plots), the bounds are still very tight for small path lengths (e.g., $d<4$), while the range \textit{[lower bound, upper bound]} increases with $d$.  In summary, the accuracy of the bounds increases for smaller $k$ or $d$.%; \blue{the explanation for this can be found in the proof of Theorem~\ref{thm:sd-path-d-k}, where the uncertainty under which we consider the inequalities, decreases when $k$ or $d$ is small.}

For $k=200$ and $d=7$ (rightmost points in bottom plots), while the simulated value lies in the middle of the two bounds in the exponential $f_{bgp}(t)$ case (Fig.~\ref{fig:bounds-exp}), it is closer to the upper bound in the uniform $f_{bgp}(t)$ case (Fig.~\ref{fig:bounds-uni}). Among all the scenarios we tested, we did not observe any tendency of the values to be closer either to the upper or lower bound. This is an indication that there is probably a limit on how much tighter bounds can be derived.

In Table~\ref{table:bounds}, we show how the times $T_{SD}$ change for increasing SDN cluster size $k$. Comparing the two cases, $d=2$ and $d=5$, we can see that the effect of the routing centralization is higher for longer paths. The simulated data-plane connectivity times decrease more and faster for $d=5$, and this is captured also by the relative changes of the theoretical bounds. 

Similar behavior is observed also in Fig.~\ref{fig:t-sd-betweenness-caida}, in simulation scenarios on the Internet topology where the SDN cluster comprises nodes with high betweenness centrality. For $k=10$, paths of length $d=3$, $d=6$, and $d=9$, see a relative decrease on the average connectivity time of about $10\%$, $20\%$, and $40\%$, respectively. The corresponding values for $k=50$ are about $25\%$, $40\%$, and $60\%$ (i.e., almost double than $k=10$), while for larger SDN cluster sizes ($k>50$) the extra gain is small.

These findings (Table~\ref{table:bounds} and Fig.~\ref{fig:t-sd-betweenness-caida}) demonstrate that ASes which have (on average) longer paths to other ASes, e.g., stub networks or small ISPs at the edge of the Internet, would see a higher benefit from routing centralization than central ASes (e.g., tier-1 ISPs) or well connected ASes such as CDNs~\cite{Chiu-One-Hop-Away-IMC-2015}. Hence, the \textit{node closeness centrality}~\cite{Newman:Networks-book} can be used as a metric to evaluate (or rank) the improvement in the performance of ASes: the lower the closeness centrality, the higher the benefit from routing centralization.

The above observation sheds light on an interesting trade-off related to which nodes participate to the SDN cluster and which nodes benefit from routing centralization. As shown in Section~\ref{sec:topo-and-sdn}, nodes with high betweenness centrality improve more the performance if they participate in the SDN cluster (see, e.g., Fig.~\ref{fig:bounds-random-betweenness-CAIDA}). However, their own gain is smaller since they are central nodes in the network (betweenness and closeness centrality are positively correlated measures). As a result, incentives -other than performance- might be needed for attracting central ASes to cooperate for routing centralization. For instance, tier-1 ISPs could deploy inter-domain centralization in order to offer new services (related to the improved BGP convergence performance) to their customers.

\begin{figure}
\centering
\subfigure[$T_{bgp}\sim exponential(\lambda=1)$]{\includegraphics[width=0.45\linewidth]{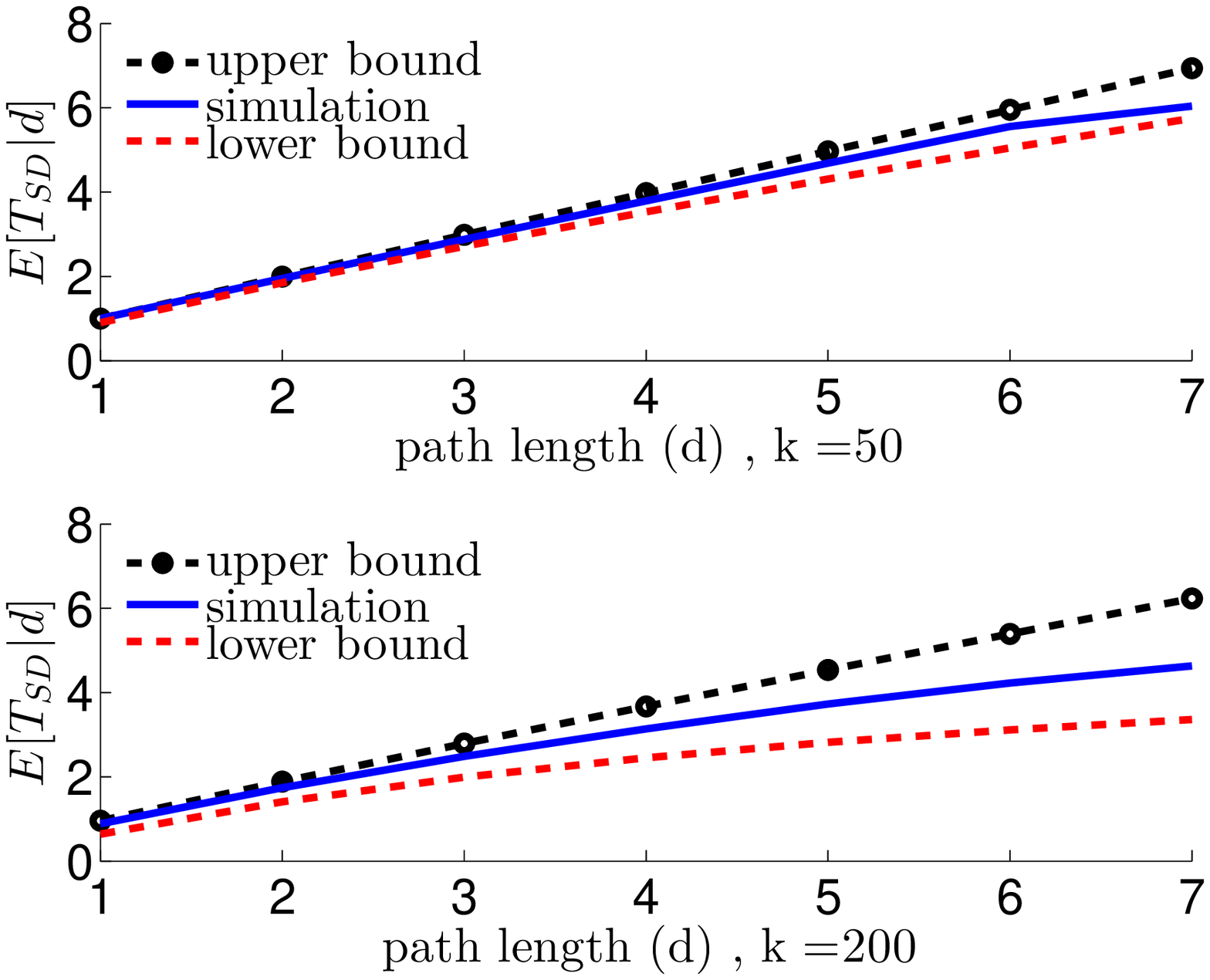}\label{fig:bounds-exp}}
\hspace{0.05\linewidth}
\subfigure[$T_{bgp}\sim uniform(0,2)$]{\includegraphics[width=0.45\linewidth]{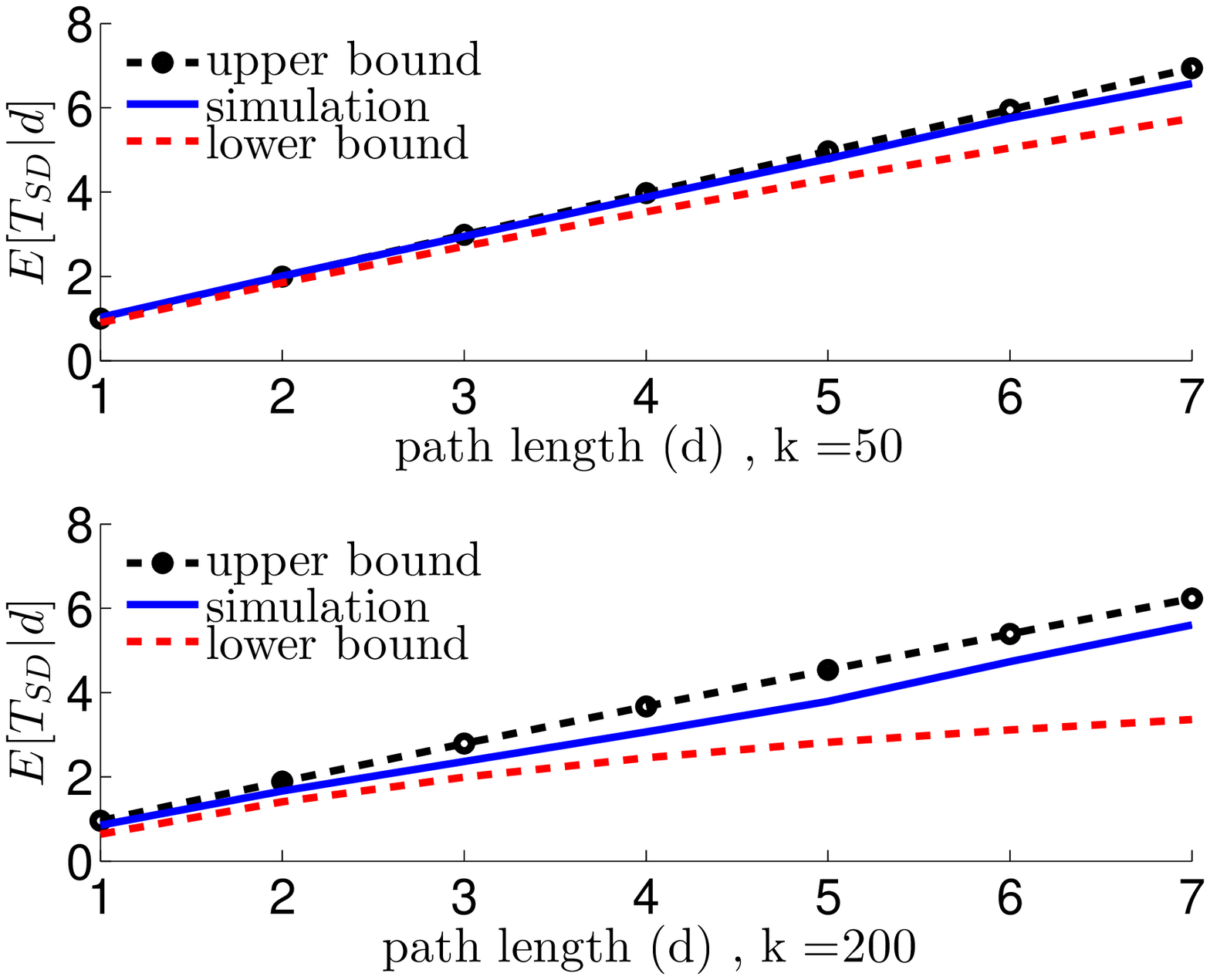}\label{fig:bounds-uni}}
\caption{Data-plane connectivity time $E[T_{SD}|d]$ (y-axis), vs. size of network cluster $k$ (x-axis). Simulation scenarios: Poisson graph network topology of size $N=1000$ and $p=0.005$, with (a) $T_{bgp}\sim exponential(\lambda=1)$ and (b) $T_{bgp}\sim uniform(0,2)$.}
\label{fig:bounds}
\end{figure}

\begin{table*}
\centering
\caption{Data-plane connectivity time normalized over the no SDN scenario, $\frac{E[T_{SD}|k]}{E[T_{SD}|k=0]}$.}
\label{table:bounds}
\begin{tabular}{|c|cccc|}
\hline
{Upper bound / \textbf{Simulation} / Lower bound}		& {$k=20$}								& {$k=50$}		& {$k=100$}		& {$k=200$}\\
\hline
{$d=2$}				& {99.9\% / \textbf{99.2\%} / 97.0\%}	& {99.6\% / \textbf{97.7\%} / 92.5\%}	& {98.6\% / \textbf{92.9\%} / 85.1\%}	& {94.4\% / \textbf{85.1\%} / 70.4\%}	\\
{$d=5$}				& {99.9\% / \textbf{97.8\%} / 94.2\%}	& {99.3\% / \textbf{93.9\%} / 86.2\%}	& {97.4\% / \textbf{86.4\%} / 74.5\%}	& {90.1\% / \textbf{75.6\%} / 56.4\%}	\\
\hline
\end{tabular}
\end{table*}

\begin{figure}
\centering
\includegraphics[width=\linewidth]{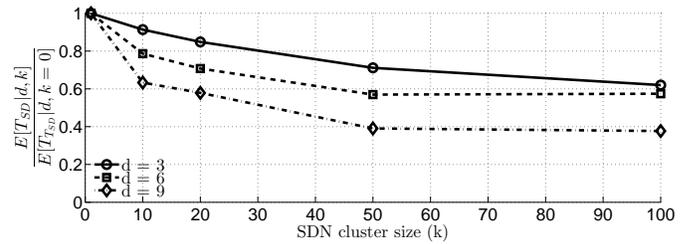}
\caption{Data-plane connectivity time normalized over the no SDN scenario $\frac{E[T_{SD}|d,k]}{E[T_{SD}|d,k=0]}$ (y-axis) vs. size of network cluster $k$ (x-axis). Curves correspond to the averages for three different path lengths (i) $d=3$, (ii) $d=6$, and (iii) $d=9$, in simulation scenarios over the Internet AS-graph, with $T_{bgp}\sim exponential(\lambda=1)$, and nodes in the SDN cluster selected with decreasing \textit{betweenness centrality}.}
\label{fig:t-sd-betweenness-caida}
\end{figure}

\section{Control-Plane Convergence}\label{sec:control-plane}
In this section we derive results for the control plane convergence time, i.e., the time needed after a routing change till \textit{every} AS in the network has updated and established the final (i.e., shortest, conforming to routing policies) paths.

The control-plane convergence time is equal to the maximum of the $T_{SD}$ times over all the SD-paths. Due to the involved order statistics, proceeding similarly to Section~\ref{sec:data-plane}, would lead to complex computations and loose bounds. Hence, in this section, we proceed to an approximate analysis that allows us to provide useful insights for the effects of routing centralization on the BGP convergence time. %is a trade-off between accuracy and analytical tractability.

% Hence, to be able to obtain an analytical understanding (which is the main goal of this study) also on the control-plane convergence, in this section, we proceed to an approximate analysis that is a trade-off between accuracy and analytical tractability.

Specifically, we first narrow the Assumption~\ref{assumption:t-bgp}, by assuming that the renewal process for the BGP update times $T_{bgp}$ is a Poisson process; this allows to study the problem using a Markovian framework. %\red{The Poisson assumption is commonly used in the literature to approximate and study propagation processes on complex networks (e.g., epidemics, mobile networks, computer viruses).} 
Our experiments and measurements in the real Internet (Appendix~\ref{sec:distr-t-bgp}), support the selection of the Poisson assumption for the times $T_{bgp}$.

\begin{assumption}[BGP updates - Poisson process]\label{assumption:t-bgp-poisson}
The times $T_{bgp}$ are iid random variables, drawn from an exponential distribution with rate $\lambda = \frac{1}{\mu_{bgp}}$, and mean value $E[T_{bgp}] =~\mu_{bgp}$.
\end{assumption}

Under Assumption~\ref{assumption:t-bgp-poisson}, we can build a \textit{transient} Markov Chain  to model the propagation of BGP updates, where each state denotes the set of nodes that have updated the paths in their RIBs. However, analysing such a Markov chain is still very complex, since the state space contains $2^{N}-1$ states, and the transition rates depend on the topology of the network, which cannot be known exactly in most practical cases.

%, set of nodes in the SDN cluster, and node that initiated the routing change. In most practical cases, the exact topology of the whole network cannot be known exactly (see also discussion of Section~\ref{}), or even if it was known, the solution could be found only through numerical calculations (due to the aforementioned complexity).

%To this end, we first consider the case of a full-mesh network, which can be described by a much simpler Markov chain, and compute the control-plane convergence time as a function of the network size $N$, SDN cluster size $k$, and rate $\lambda$. Then, we generalize the results, and derive approximations for sparse topologies, which are of higher practical interest.

To this end, we first consider the case of a full-mesh network (a common approach in related literature~\cite{Kotronis-Routing-Centralization-ComNets-2015,Labozitz-Delayed-convergence-CCR-2000,convergence-properties-BGP_ComNets-2011}), which can be described by a much simpler Markov chain, and compute the control-plane convergence time as a function of the network size $N$, SDN cluster size $k$, and rate $\lambda$ (Section~\ref{sec:full-mesh}). Then, we generalize the results, and derive approximations for sparse topologies, which are of higher practical interest (Section~\ref{sec:poisson}). Simulation results show that the insights stemming from our analysis are valid also for the (much more complex) Internet AS-graph (Section~\ref{sec:control-plane-validation}).

\subsection{Analysis: Full-Mesh Topology}\label{sec:full-mesh}
%Analysis of BGP convergence in full-mesh graphs is common in related literature~\cite{Kotronis-Routing-Centralization-ComNets-2015,Labozitz-Delayed-convergence-CCR-2000,convergence-properties-BGP_ComNets-2011}, since it makes easier to obtain initial insights, and is also motivated by the recent trends on the Internet evolution (i.e., peering connections increase, flattening of Internet~\cite{Gregori-Impact-IXPs-ComCom-2011}).

In a full-mesh network, every pair of nodes has a direct connection, and, thus, the shortest path (i.e., BGP path) to each node is the direct path of size $d=1$. Hence, every node receives the BGP update from the source node. Moreover, since all nodes in the SDN cluster are informed the time any of them receives the BGP update ($T_{sdn}\ll T_{bgp}$, or $T_{sdn}\rightarrow 0$), the SDN cluster can be considered as a single node.

As a result, a Markov Chain as this in Fig.~\ref{fig:mc-steps} can be used to model the propagation of BGP updates. Each time a node (a single AS or the SDN cluster) receives the BGP update, the Markov chain moves to the next state. %\red{We start from the moment/state (time $t=0$ / state $0$) just before the routing change takes place. If the routing change is initiated by a node in the SDN cluster, the source node is the SDN cluster and $k$ nodes have the updated BGP paths. Otherwise, a single node (the source) knows about the routing change. We say that there is control-plane convergence, and denote it with the state $C$, when all nodes have the updated paths in their RIBs.} 
We start from the moment/state (time $t=0$ / state $0$) just before the routing change takes place. Control-plane convergence is achieved at state $C$, when all nodes have the updated paths in their RIBs.

To calculate the transition rates $\lambda_{i}^{'}$, we first define the following quantities.
\begin{definition}[bgp-eligible nodes \& bgp-degree]\label{def:bgp-degree}~\\
$-$ A \texttt{bgp-eligible} node is a node the (a) has not received the BGP update, and (b) lies on a BGP (shortest) path where the previous node has the updated route in its RIB.~\\ %We denote, at a step $i$, the set of bgp-eligible nodes as $\mathcal{D}(i)$.
$-$ The \texttt{bgp-degree} at step $i$, $D(i)$, is the number of nodes that are bgp-eligible nodes.%, i.e., $D(i) = |\mathcal{D}(i)|$.
\end{definition}
Under the above definition, the time to move from a step/state $i$ to the next step/state, is the time needed till the \textit{first} of the bgp-eligible nodes receives the update. Under Assumption~\ref{assumption:t-bgp-poisson}, it follows that this time is the minimum of $D(i)$ iid random variables exponentially distributed with rate $\lambda$. Therefore the transition time is also exponentially distributed with rate (i.e., the transition rate)
\begin{equation}\label{eq:transition-rate-lambda-prime}
\lambda_{i}^{'} = \lambda\cdot D(i)
\end{equation}

Now, in a full-mesh network, bgp-eligible nodes are all the nodes that have not received the BGP update (since all nodes are directly connected to the source node). We denote as $n(i)$ the number of nodes that have received the BGP update at step $i$. From the above discussion it follows $n(i)$ depends on which step the SDN cluster received the BGP update. Denoting as $x$ the state/step that the first node in the SDN cluster receives the BGP update, we can write
\begin{equation}\label{eq:n(i)}
n(i|x) = \left\{
\begin{tabular}{ll}
$i$	& $, i\leq x$ \\
$i+k-1$	& $, i>x$
\end{tabular}
\right.
\end{equation}
and the bgp-degree is easily shown to be given by Lemma~\ref{thm:Dix-full-mesh}.
\begin{lemma}\label{thm:Dix-full-mesh}
The \bgp $\Dix$, $i\in[1,N-k], x\in[0,N-k]$, in a full-mesh network topology is given by
\begin{equation}
\Dix = N-n(i|x)
\end{equation}
\end{lemma}

Up to this point, we have calculated the transition rates of the Markov chain of Fig.~\ref{fig:mc-steps} conditionally on $x$ (see, \eq{eq:transition-rate-lambda-prime} and Lemma~\ref{thm:Dix-full-mesh}). To compute the control-plane convergence time, we need also the probabilities $P_{sdn}(x)$ that the SDN cluster receives the BGP update at step $x$. In the following lemma, we derive the expression for the probabilities $P_{sdn}(x)$.
%To compute the control-plane convergence time, we need to calculate also the probabilities $P_{sdn}(x)$ that the SDN cluster receives the BGP update at step $x$.
%Proceeding recursively, we derive the following lemma that gives the probability $P_{sdn}(x)$.

\begin{lemma}\label{thm:P-sdn}
The probability that the SDN cluster receives the update at step $x$ is given by
\begin{equation}\label{eq:P-sdn}
P_{sdn}(x) = \frac{k}{N-x}\cdot \prod_{j=0}^{x-1}\left(1-\frac{k}{N-j}\right)
\end{equation}
\end{lemma}
\begin{proof}
The proof is given in Appendix~\ref{sec:proof-of-thm-P-sdn}.
\end{proof}

%
%\begin{figure}
%\subfigure[Markov Chain (number of nodes)]{\includegraphics[width=\linewidth]{./figures/MC_nb_of_nodes1.eps}\label{fig:mc-nodes}}
%\subfigure[Markov Chain (number of transitions, or \textit{steps})]{\includegraphics[width=\linewidth]{./figures/MC_steps.eps}\label{fig:mc-steps}}
%\caption{Markov Chains where the states correspond to (a) the number of nodes that have updated BGP routes, and (b) the number of transitions, or \textit{steps}, of the BGP update dissemination process.}
%\label{fig:markov-chains}
%\end{figure}

\begin{figure}
\includegraphics[width=\linewidth]{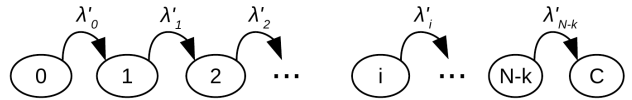}
\caption{Markov Chain for the BGP update dissemination process.}\label{fig:mc-steps}
\end{figure}

Now, using Lemmas~\ref{thm:Dix-full-mesh} and~\ref{thm:P-sdn}, we proceed and derive the following result for the distribution of the control-plane convergence time $T_{c}$. Specifically, Lemma~\ref{thm:MGF-Tc} gives a closed form expression for the moment generating function (MGF)\footnote{
We remind that the MGF of a random variable $X$ is defined as $M_{X}(\theta) = E[e^{\theta\cdot X}]$, $\theta\in\mathbb{R}$, and completely characterizes a random variable (equivalently to its distribution), since all the moments of $X$ can be calculated from its MGF.% as $E[X^{n}] = \left. \frac{d^{n}M_{X}(\theta)}{(d\theta)^{n}}\right|_{\theta=0}$
} of the time $T_{c}$.

%We remind that the MGF of a random variable $X$ is defined as
%\begin{equation}
%M_{X}(\theta) = E[e^{\theta\cdot X}]~~,~~~~\theta\in\mathbb{R}
%\end{equation}
%and completely characterizes a random variable (it is considered equivalent to its distribution), since all the moments of the variable $X$ can be calculated from its MGF as
%\begin{equation}\label{eq:MGF-moments}
%E[X^{n}] = \left. \frac{d^{n}M_{X}(\theta)}{(d\theta)^{n}}\right|_{\theta=0}
%\end{equation}

\begin{lemma}\label{thm:MGF-Tc}
The moment generating function (MGF) $M_{T_{c}}(\theta)$ of the BGP convergence time $T_{c}$ is given by
\begin{equation}
M_{T_{c}}(\theta) = \sum_{x=0}^{N-k}\prod_{i=1}^{N-k}\left(1-\frac{\theta}{\lambda\cdot D(i|x)}\right)^{-1} \cdot P_{sdn}(x)
\end{equation}
\end{lemma}
\begin{proof}
The proof is given in Appendix~\ref{sec:proof-of-thm-MGF-Tc}.
\end{proof}

Using the above lemma, and applying the property
\begin{equation}\label{eq:MGF-moments}
E[X^{n}] = \left. \frac{d^{n}M_{X}(\theta)}{(d\theta)^{n}}\right|_{\theta=0}
\end{equation}
we can calculate the moments of $T_{c}$. The following theorem gives the mean value (first moment) %\footnote{Higher moments can be calculated in a similar way.}
 of $T_{c}$ as a function of $D(i|x)$ (Lemma~\ref{thm:Dix-full-mesh}) and $P_{sdn}(x)$ (Lemma~\ref{thm:P-sdn}), or, equivalently, as a function of the parameters $N$, $k$, and $\lambda$. %\underline{Note:} Higher moments can be calculated in a similar way, and result to closed form expressions as well. %Here, we focus on the first two moments that are of higher practical interest.
\begin{theorem}\label{thm:expectation-and-variance-Tc}
The expectation of the BGP convergence time $T_{c}$ is
\begin{equation}
E[T_{c}] = \frac{1}{\lambda}\cdot\sum_{x=0}^{N-k}\sum_{i=1}^{N-k}\frac{1}{ D(i|x)}\cdot P_{sdn}(x)
\end{equation}
%The expectation (i.e., first moment) and second moment of the BGP convergence time $T_{c}$ is
%\begin{align}
%E[T_{c}] 	& = \frac{1}{\lambda}\cdot\sum_{x=0}^{N-k}\sum_{i=1}^{N-k}\frac{1}{ D(i|x)}\cdot P_{sdn}(x)	\\
%E[T_{c}^{2}] & = \frac{1}{\lambda^{2}}\sum_{x=0}^{N-k}\left(\sum_{i=1}^{N-k}\frac{1}{\left(D(i|x)\right)^{2}}
%				+ \left(\sum_{i=1}^{N-k}\frac{1}{D(i|x)} \right)^{2}\right)\cdot P_{sdn}(x)
%\end{align}
\end{theorem}

The methodology in the proof of Lemma~\ref{thm:MGF-Tc} can be applied to derive useful expressions for other quantities that are of practical interest, and allow us to obtain a better understanding of the effects of routing centralization on control-plane convergence. For example, the following corollary quantifies the speed of the control-plane convergence process.%(whose proof is omitted due to space limitation)

%\begin{definition}
%$\ell$-Partial BGP Convergence Time, $T_{\ell}$, is the time needed till $\ell$ ($\ell\leq N$) ASes have the final BGP updates.
%\end{definition}

\begin{corollary}\label{thm:l-partial-Tc}
The expectation of the $\ell$-Partial BGP Convergence Time, $T_{\ell}$, i.e., the time needed till $\ell$ ($\ell\leq N$) nodes have the final BGP updates,is given by
\begin{equation}
E[T_{\ell}] = \frac{1}{\lambda}\cdot \sum_{x=0}^{N-k} \sum_{i=1}^{M(\ell,x)}\frac{1}{D(i|x)}\cdot P_{sdn}(x)
\end{equation}
where
\begin{equation}
M(\ell,x) = \left\{
\begin{tabular}{ll}
$\ell-1$		&		, $~~~~~0<\ell\leq x+1$	\\
$x$			&		, $x+1<\ell\leq x+k$	\\
$\ell-k$		&		, $x+k<\ell\leq N$	\\
\end{tabular}
\right.
\end{equation}
\end{corollary}

%%% REMOVED COROLLARY %%%mew 
%
%\begin{corollary}\label{thm:convergenve-in-out-cluster}
%The expectation of the BGP convergence time for a routing change initiated by an AS \underline{in} the SDN cluster is
%\begin{equation}\label{eq:Tc-AS-in-SDN}
%E[T_{c}|AS\in SDN] = \frac{1}{\lambda}\cdot \sum_{i=1}^{N-k} \frac{1}{D(i|0)}
%\end{equation} 
%and for a routing change initiated by an AS \underline{outside} the SDN cluster is
%\begin{align}\label{eq:Tc-AS-notin-SDN}
%E[T_{c}|AS\notin SDN] = \frac{N\cdot E[T_{c}] - k\cdot E[T_{c}|AS\in SDN]}{N-k}
%\end{align} 
%%or
%%\begin{align}
%%E[T_{c}|AS\notin SDN] = \frac{1}{\lambda}\cdot \sum_{x=1}^{N-k}\sum_{i=1}^{N-k} \frac{1}{D(i|x)}\cdot P_{sdn}(x)
%%\end{align} 
%\end{corollary}

\subsection{Analysis: Sparse Topologies}\label{sec:poisson}
%As mentioned earlier, computing the control-plane convergence for an arbitrary topology is very complex. For instance, applying the methodology of Section~\ref{sec:full-mesh}, in \eq{eq:mgf-conditional-expectation} the terms in the product are not independent, since the set of bgp-eligible nodes at a step $i$ depends on the exact paths $\mathcal{P}$ that the BGP updates have been propagated. Hence, the expectation needs to be taken over all $S\in\mathcal{P}$ (with $|\mathcal{P}|\sim O\left(2^{N}\right)$), and we need to keep track of all $D(i|x,S\in\mathcal{P})$ and $P_{sdn}(x|S\in\mathcal{P})$.

%To avoid an intractable exact analysis, in this section, we derive approximations for sparse networks. To this end, we assume a Poisson (or, Erdos-Renyi) random graph $G(N,p)$, to capture the sparseness of a topology. However, as we show in the validation Section~\ref{sec:control-plane-validation}, our results describe well effects of routing centralization also in more generic/realistic topologies, like power-law or small-world graphs.
As mentioned earlier, computing the control-plane convergence for an arbitrary topology is very complex. For instance, applying the methodology of Section~\ref{sec:full-mesh}, the set of bgp-eligible nodes at a step $i$ depends on the exact paths $\mathcal{P}$ that the BGP updates have been propagated. Hence, we need to consider all $S\in\mathcal{P}$ (with $|\mathcal{P}|\sim O\left(2^{N}\right)$), and we need to keep track of all $D(i|x,S\in\mathcal{P})$ and $P_{sdn}(x|S\in\mathcal{P})$. However, approximating sparse topologies with a Poisson (or, Erdos-Renyi) random graph $G(N,p)$, we derive expressions for the BGP convergence time in the following result. As we show in the validation Section~\ref{sec:control-plane-validation}, our approximations describe well effects of routing centralization also in more generic/realistic topologies, like power-law graphs or the Internet AS-graph.

\begin{result}\label{thm:Dix-poisson} Lemma~\ref{thm:MGF-Tc}, Theorem~\ref{thm:expectation-and-variance-Tc}, and Corollary~\ref{thm:l-partial-Tc}, with $E[\Dix]$ (instead of $D(i|x)$), approximate the control-plane convergence time in a Poisson graph network topology; where $E[\Dix]$ is the expectation of the \bgp $\Dix$, $i\in[1,N-k], x\in[0,N-k]$, in a Poisson graph 
\begin{equation}
E[\Dix] = \left(N-n(i|x)\right) \cdot \left(1-(1-p)^{n(i|x)}\right)
\end{equation}
\end{result}
\begin{proof}
The proof is given in Appendix~\ref{sec:proof-of-thm-Dix-poisson}.
\end{proof}

\subsection{Simulation Results and Implications}\label{sec:control-plane-validation}
We evaluate the accuracy of our theoretical results in various simulation scenarios, including also scenarios where the assumptions for (i) exponential $f_{bgp}(t)$, and (ii) full-mesh or Poisson graph networks, do not hold. %We discuss the main findings and present a subset of results.

In scenarios of full-mesh networks, where the times $T_{bgp}$ are exponentially distributed, our theoretical expressions of Section~\ref{sec:full-mesh} predict the simulation results for the expected convergence time $E[T_{c}]$ with very high accuracy. 

For the validation of the theoretical expressions in sparse networks (Section~\ref{sec:poisson}), we simulate various sparse topologies, like Poisson, Barabasi-Albert (power low), and Newman-Watts-Strogatz (small world) graphs. Although the theoretical results are derived under the Poisson graph assumption, our simulations show that they can predict the performance with similar accuracy in the all the topologies we tested. 

%\red{Moreover, a main observation is that the accuracy of the approximation of Result~\ref{thm:Dix-poisson} is higher for well-connected (sparse) graphs, with (shortest) paths of smaller lengths. This is in accordance with the findings of Section~\ref{sec:data-plane}, where we show that the variability of the data-plane connectivity time increases with the path length.}

In Fig.~\ref{fig:Tell-vs-k} we present a representative subset of our results that demonstrate how the routing centralization can decrease the BGP convergence time. We plot the partial convergence time, normalized over the scenario without centralization, i.e. ,$\frac{E[T_{\ell}|k]}{E[T_{\ell}|k=0]}$. We consider three cases, $\ell=100$ (or, $0.1\cdot N$) in Fig.~\ref{fig:ell-100}, $\ell=500$ (or, $0.5\cdot N$) in Fig.~\ref{fig:ell-500}, and $\ell=N=1000$ that corresponds to the control-plane convergence in Fig.~\ref{fig:ell-N}. 

A first observation is that our results can capture well the relative changes\footnote{The accuracy of the theoretical results (approximations), when we consider the \textit{actual} -not normalized- values, is lower.} in the (partial) convergence time, not only for scenarios with exponential $f_{bgp}(t)$ (as we assume in our analysis), but also for scenarios with uniform $f_{bgp}(t)$.

In Fig.~\ref{fig:ell-N}, we can see that the control-plane convergence time does not significantly improve as the SDN cluster size $k$ increases. For instance, even for $k=500$ (i.e., $50\%$ of the nodes belong to the SDN cluster), the decrease in the convergence time is less than $30\%$. This comes to verify the results of~\cite{Kotronis-Routing-Centralization-ComNets-2015}, which showed that significant gains can be achieved only for high values ($>50\%$) of SDN penetration.

However, when it comes to the partial control-plane convergence (Figs.~\ref{fig:ell-100} and~\ref{fig:ell-500}), the effects of routing centralization are higher. The time needed till $10\%$  of the nodes ($\ell=100$ - Fig.~\ref{fig:ell-100}) to receive the updated routing information, decreases quickly; e.g., to $0.5$ of its no-SDN ($k=0$) value, only with $k=100$ nodes ($10\%$) participating in the SDN cluster. 

This reveals an important aspect, relating to the effects of routing centralization, which has not been shown in previous works (e.g.,~\cite{Kotronis-Routing-Centralization-ComNets-2015}): although the control-plane convergence can significantly improve only if a high percentage ($>50\%$) of nodes cooperate, we can have very large gains in the \textit{partial convergence} even with small sizes of SDN clusters.

In Fig.~\ref{fig:ell-caida-betweenness} we present simulation results on the Internet AS-graph\footnote{For scalability issues, we did not consider here stub ASes and ASes with less than 3 neighbors, resulting in a reduced Internet graph with $N=11527$.}, where the top betweenness centrality nodes form the SDN cluster. Despite the fact that the simulated scenario deviates from our assumptions, our main theoretical findings are still valid: centralization can significantly accelerate the connectivity time with a large percentage of ASes (i.e., $\ell$-partial convergence, see, e.g., curves for $\ell = 0.1\cdot N$ and $\ell = 0.5\cdot N$), while the time needed till every AS has received the updated routes (i.e., total convergence $E[T_{c}]$) improves more slowly with the SDN cluster size $k$. Moreover, we can see that the efficiency of inter-domain centralization is quite impressing; with only $k=50$ central nodes in the SDN cluster, the time needed to establish updated paths with half of the Internet nodes ($\ell=0.5\cdot N$) is $50\%$ less than in the case without centralization.

\begin{figure}
\centering
\subfigure[$\ell = 100$]{\includegraphics[width=0.9\linewidth]{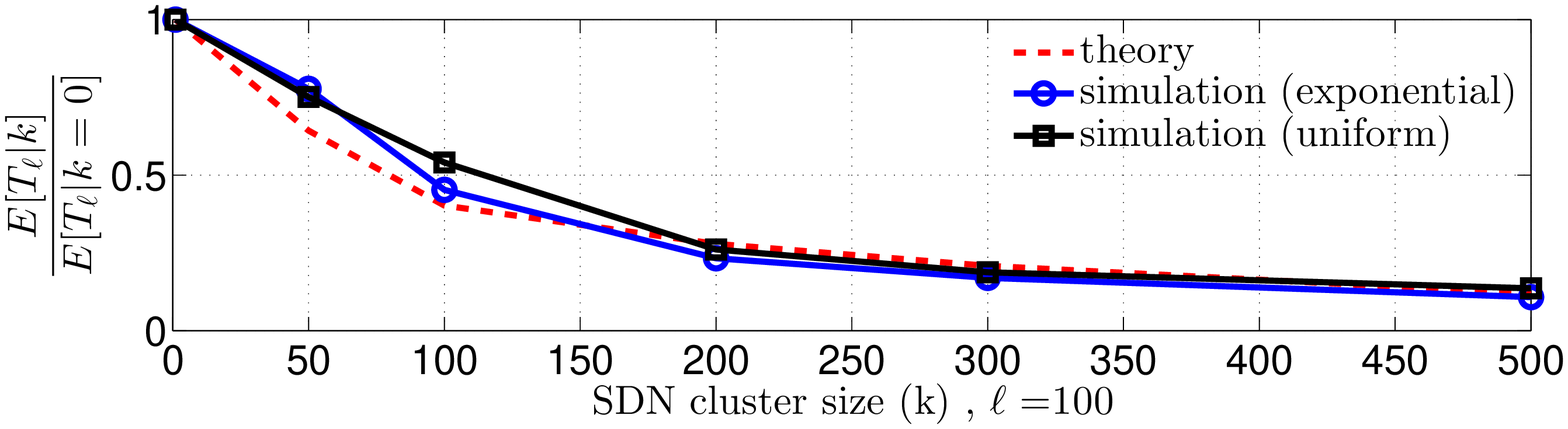}\label{fig:ell-100}}
\subfigure[$\ell = 500$]{\includegraphics[width=0.9\linewidth]{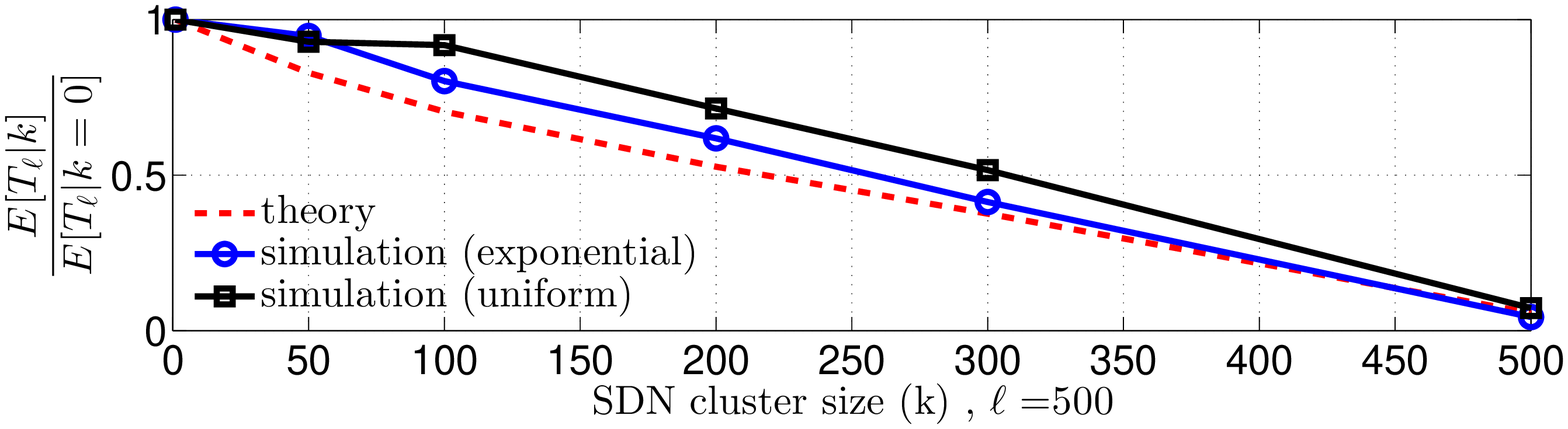}\label{fig:ell-500}}
\subfigure[$\ell = 1000 (=N)$]{\includegraphics[width=0.9\linewidth]{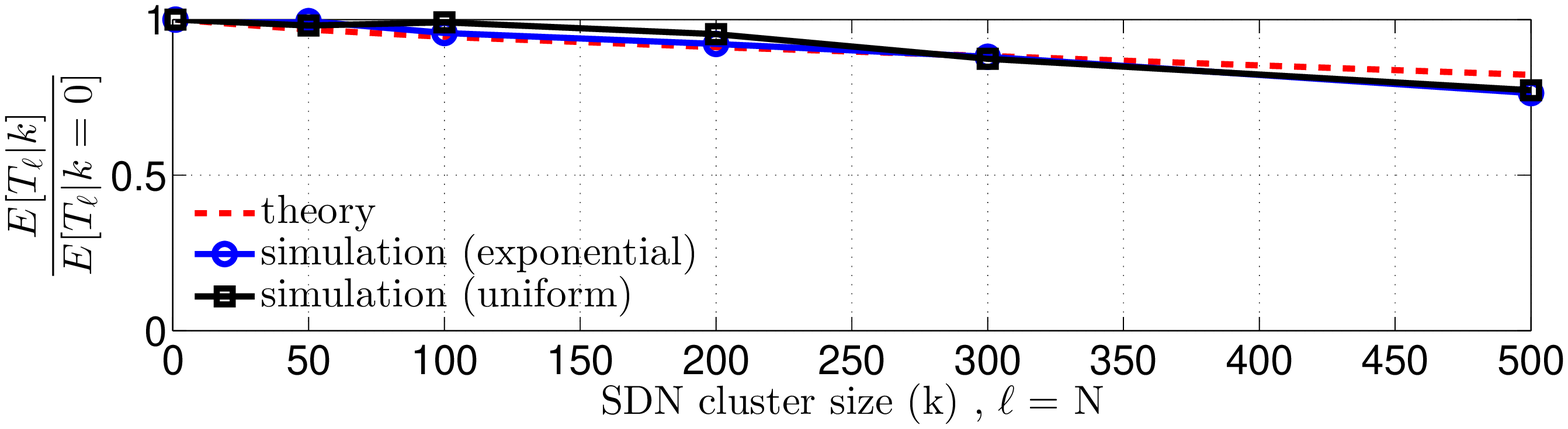}\label{fig:ell-N}}
\caption{Partial convergence time, normalized over the no SDN scenario, $\frac{E[T_{\ell}|k]}{E[T_{\ell}|k=0]}$ (y-axis), vs. size of SDN cluster $k$ (x-axis). Simulation scenarios: Barabasi-Albert topology with $N=1000$ and average node degree $10$; $T_{bgp}\sim exponential(\lambda=1)$ (black line - squares) and $T_{bgp}\sim uniform(0,2)$ (blue line - circles).}
\label{fig:Tell-vs-k}
\end{figure}

\begin{figure}
\centering
\includegraphics[width=\linewidth]{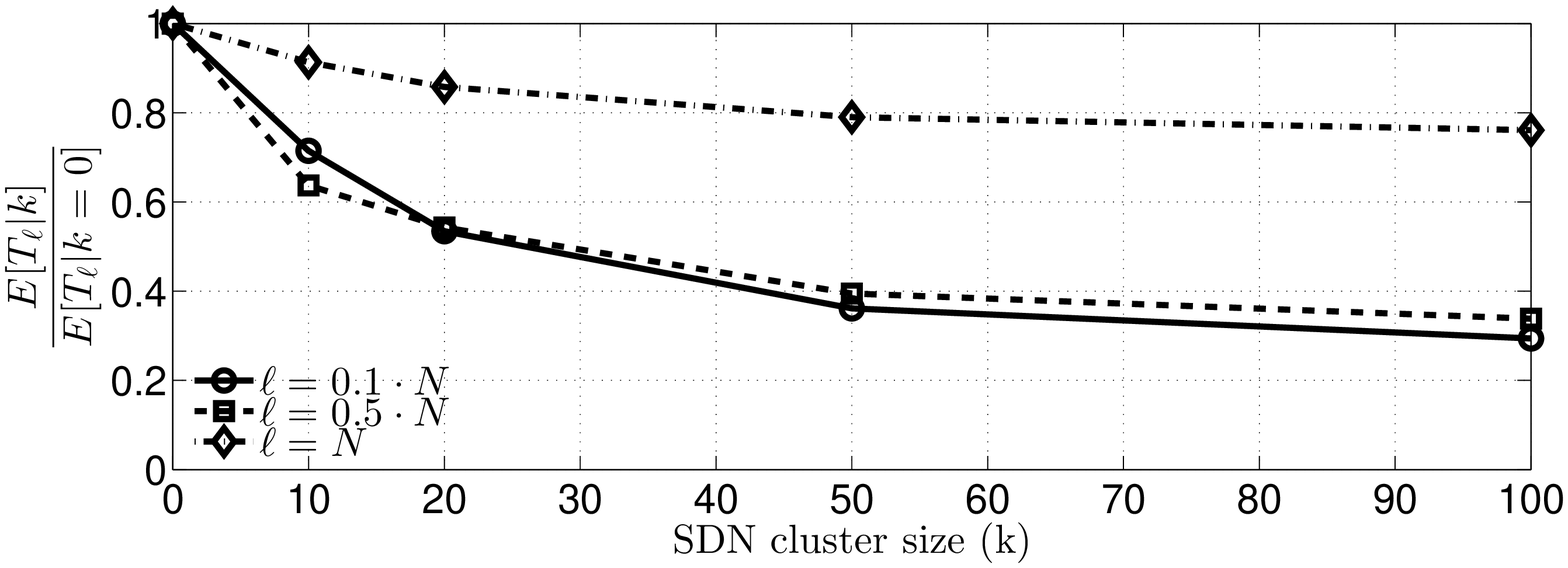}
\caption{Partial convergence time, normalized over the no SDN scenario, $\frac{E[T_{\ell}|k]}{E[T_{\ell}|k=0]}$ (y-axis), vs. size of SDN cluster $k$ (x-axis). Simulation scenarios on the Internet AS-graph. Nodes in the SDN cluster are selected with decreasing \textit{betweenness centrality}.}
\label{fig:ell-caida-betweenness}
\end{figure}

%
%\begin{figure}
%\centering
%\subfigure[SDN cluster: random]{\includegraphics[width=\linewidth]{./figures/fig_T_ell_CAIDA_dataset.eps}\label{fig:ell-caida-random}}
%\subfigure[SDN cluster: betweenness]{\includegraphics[width=\linewidth]{./figures/fig_T_ell_CAIDA_dataset_betweenness.eps}\label{fig:ell-caida-betweenness}}
%\caption{Partial convergence time, normalized over the no SDN scenario, $\frac{E[T_{\ell}|k]}{E[T_{\ell}|k=0]}$ (y-axis), vs. size of SDN cluster $k$ (x-axis). Simulation scenarios on the Internet AS-graph. Nodes in the SDN cluster are selected (i) \textit{randomly} and (ii) with decreasing \textit{betweenness centrality}.}
%\label{fig:Tell-vs-k-caida}
%\end{figure}

\section{Related Work}
Inter-domain SDN is a new research area that attracts increasing attention~\cite{Gupta-SDX-CCR-2014, Kotronis-CXP-SOSR-2016,Lin-Seamless-Internetworking-Demo-Sigcomm-2013,Thai-Decoupling-BGP-Conext-2012,Rothenberg-Revisiting-RCP-HotSDN-2012,Bennesby-Innovating-IDrouting-AINA-2014,Kotronis-Routing-Centralization-ComNets-2015}. In~\cite{Gupta-SDX-CCR-2014} authors propose and implement SDX, a software-defined component for IXPs, which increases the capabilities on routing control. Another IXP-based system that enables novel services for establishing QoS route paths is described in~\cite{Kotronis-CXP-SOSR-2016}. In~\cite{Lin-Seamless-Internetworking-Demo-Sigcomm-2013} a solution for incremental deployment of inter-domain SDN, which is seamless to traditional IP networks, is proposed, and~\cite{Thai-Decoupling-BGP-Conext-2012} contributes in this direction by proposing an SDN-based methodology for decoupling BGP policy control from routing. \cite{Rothenberg-Revisiting-RCP-HotSDN-2012} proposes an SDN-based architecture to enhance inter-domain routing, and~\cite{Bennesby-Innovating-IDrouting-AINA-2014} proposes an component to enable inter-domain SDN. Finally, authors in~\cite{Kotronis-Routing-Centralization-ComNets-2015} build a realistic emulator, and use it to investigate the effects of routing centralization on BGP convergence time.

%outsourcing NF~\cite{Gibb-Outsourcing-NF-HotSDN-2012,Lakshminarayanan-RaaS-report-2006}

The slow convergence of BGP has been extensively studied through measurements in~\cite{Labozitz-Delayed-convergence-CCR-2000,Kushman-Can-Hear-CCR-2007,Oliveira-Quantifying-Path-Exploration-ToN-2009}. It has been shown that BGP can take several minutes to converge after a routing change, and this can cause severe packet losses~\cite{Labozitz-Delayed-convergence-CCR-2000} and performance degradation~\cite{Kushman-Can-Hear-CCR-2007}. 

Finally, analytic approaches for the BGP convergence can be found in~\cite{convergence-properties-BGP_ComNets-2011,Labozitz-Delayed-convergence-CCR-2000,stability-inter-domain-Infocom-2009}. In~\cite{convergence-properties-BGP_ComNets-2011}, a probabilistic model and automata theory is used to study the BGP convergence (probability of convergence, and convergence time). \cite{Labozitz-Delayed-convergence-CCR-2000} studies analytically the BGP convergence with respect to the number of exchanged messages, while~\cite{stability-inter-domain-Infocom-2009} performs a worst-case analysis of BGP convergence

%based on automata theory and numerical methods and asymptotic results for convergence time 

%investigate how ASes can collaborate to improve the security of BGP, e.g., based on routing (BGP paths) information~\cite{collaborative-bgp-security-IFIPnet-2016}

%~\cite{BGP-high-SDN-techrep}

%uses a probabilistic model and automata theory to study the BGP convergence (probability of convergence, and convergence time) ~\cite{convergence-properties-BGP_ComNets-2011} based on automata theory and numerical methods and asymptotic results for convergence time 
%
%~\cite{Labozitz-Delayed-convergence-CCR-2000} also studies analytically the BGP convergence time, but as a function of the number of exchanged messages.
%
%or worst-case analysis~\cite{stability-inter-domain-Infocom-2009} of BGP convergence

\section{Conclusion}\label{sec:conclusion}
In this paper, we analytically studied the effects of inter-domain SDN on the time needed for establishing connectivity and convergence after a routing change. We proposed a probabilistic model, and derived results for the expected data-plane connectivity time (lower/upper bounds) and control-plane convergence time (exact predictions and approximations).

Our results can be used to quickly evaluate the effects of different network parameters, like network size, topology, path lengths, or number of SDN nodes, on the routing performance. Hence, they can complement previous system-oriented studies and facilitate future research. Moreover, our methodology and results can be a useful tool for studying important problems relating to routing changes in the Internet. Finally, they can be applied in practical design problems, like selecting the nodes to participate in the SDN cluster based on performance criteria (i.e., which node can have the highest impact), or for network economics purposes (e.g., detecting the potential incentives for an AS to participate in inter-domain routing centralization).

%
%the performance for example our theoretical predictions show that that the control-plane convergence for not very sparse graphs (e.g., Poisson graph with edge probability $p-0.5$) has negligible differences from a full-mesh topology. A similar observation has been made in~\cite{} where rrealistic emulations have been used. Hence a researcher, with our theory at hands, could spot the most appropriate scenarios to test their system or perform a realistic (and resource demanding) validation (e.g., emulations with real software or sth else)
%
%
%selection policies for the SDN cluster. This is not only once, since routing centralization can be used in a per application basis, and only for specific purposes, and not as a continuous solution. E.g., cooperation of different ASes to automatically mitigate BGP prefix hijacking attacks, cooperation of ASes for QoS path-stitching services, etc. Our results can be used both for selecting nodes based on performance, i.e., which node can have the highest impact, or for network economics purposes, e.g., what are the incentives for an AS to participate in an SDN cluster, i.e., how much its own performance will be improved or how much it can gain by its peers, etc.

%\section{Extras}
%\input{Extras}

%\bibliographystyle{IEEEbib}
\bibliographystyle{ieeetr}
%\bibliographystyle{IEEEtran}	

%\bibliography{references}

%\newpage
\appendices
\section{Distribution of BGP update times $T_{bgp}$}\label{sec:distr-t-bgp}
To investigate if and how well our modeling assumptions can describe the BGP update times in the Internet, we compare them against real measurement data. 

We conducted experiments in the Internet using the PEERING testbed~\cite{Schlinker-PEERING-HotNets-2014}, which owns IP prefixes and ASNs, peers with networks in different locations around the world, and allows users to make real BGP announcements. In our experiments/measurements, we follow a similar methodology as in~\cite{ARTEMIS-Demo-Sigcomm-2016}: we (i) announce a /24 prefix from a site of the PEERING testbed, and (ii) use publicly available control-plane monitoring services (route collectors and looking glass servers)~\cite{bgpmon, ripe-ris,periscope} to measure the time needed till different ASes receive our announcements.

We collected BGP updates, as seen from the monitors, from $M=40$ ASes. We repeated the experiments $14$ times; each time making a BGP announcement either from the PEERING site at an IXP at Amsterdam (NL), or at an ISP at Los Angeles (US). From each received BGP update $i$, we consider (a) $T_{SD}(i)$, the time needed till the BGP update $i$ received by the monitor (i.e., timestamp of the BGP update $i$ minus the timestamp of our BGP announcement), and (b) $d(i)$, the length of the AS-path included in the BGP update $i$. 

We group the times $T_{SD}(i)$ by the respective path lengths $d(i)$, and plot the distribution (CCDF) of the measured times $T_{SD}$ in Fig.~\ref{fig:measurements-poisson-assumption} for two example cases with $d=2$ and $d=5$.

Then, we fit the real data with a distribution $f_{bgp}(t)$ (cf. Section~\ref{sec:model}), where we select $f_{bgp}(t)\sim exponential(\lambda)$ in order to test the validity of (the stronger) Assumption~\ref{assumption:t-bgp-poisson}. We estimate the \textit{average} BGP update time from the measured data as $\hat{E}[T_{bgp}] = \frac{\sum_{i}T_{SD}(i)}{\sum_{i} d(i)}$ and set the rate $\lambda = \frac{1}{\hat{E}[T_{bgp}]}$. 

We generate from $f_{bgp}(t)$ a large number of times $T_{SD}$ for paths of length $d=2$ and $d=5$, calculate their CCDFs, and compare them against the real data in Fig.~\ref{fig:measurements-poisson-assumption}. As we can observe, there is a good match between the generated and real data. This indicates that Assumption~\ref{assumption:t-bgp-poisson} is a realistic and reasonable assumption, and, thus, emphasizes the practicality of our theoretical and simulation findings in real settings

%and thus the insights stemming from our analysis and simulations will remain valid in real settings.

%\begin{equation}
%\hat{E}[T_{bgp}] = \frac{\sum_{i}T_{SD}(i)}{\sum_{i} d(i)}
%\end{equation}
%and set the rate $\lambda = \frac{1}{\hat{E}[T_{bgp}]}$.

\begin{figure}
\centering
\subfigure[$d = 2$]{\includegraphics[width=0.49\linewidth]{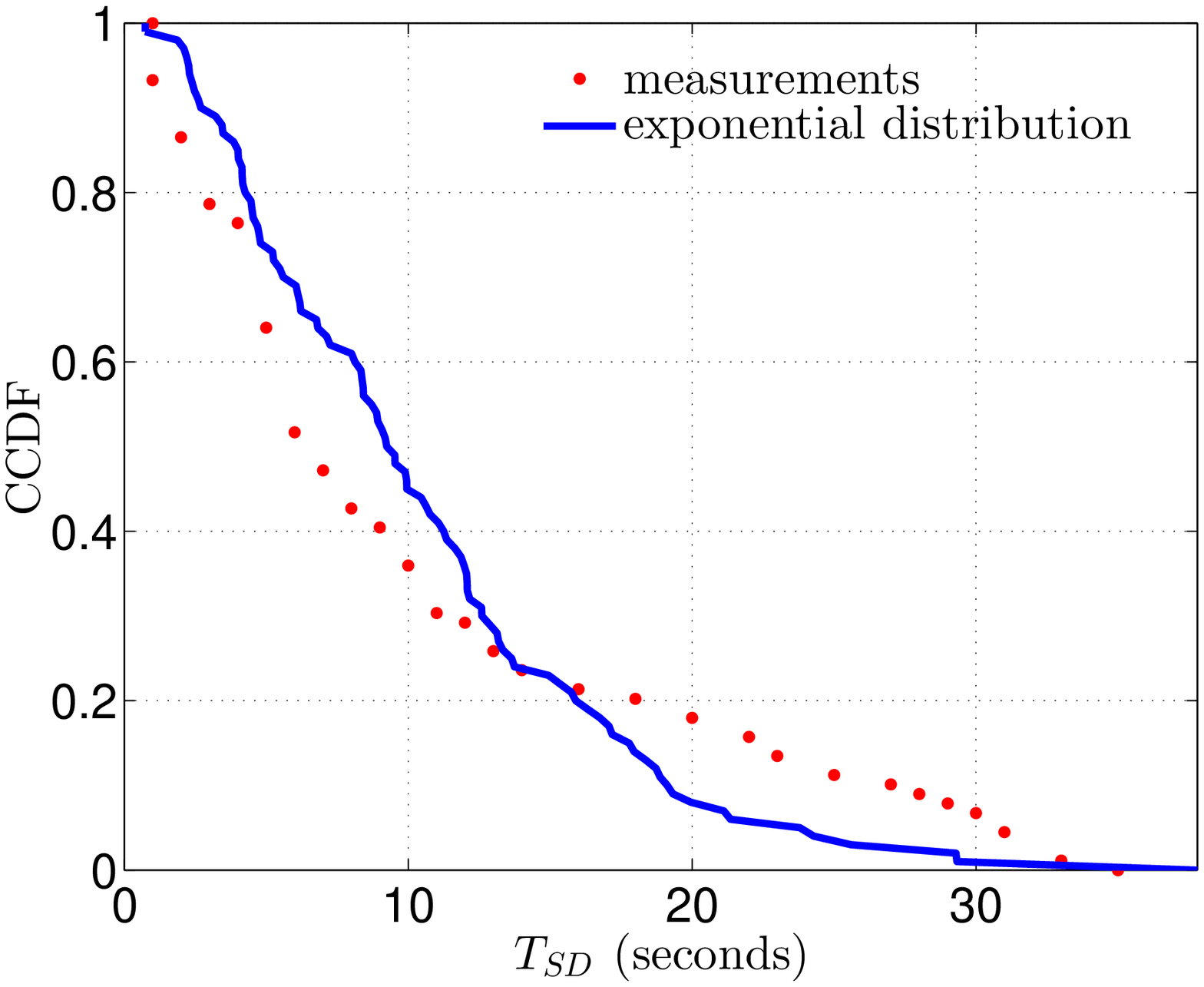}\label{fig:measurements-poisson-assumption-path-length-2}}
\subfigure[$d = 5$]{\includegraphics[width=0.49\linewidth]{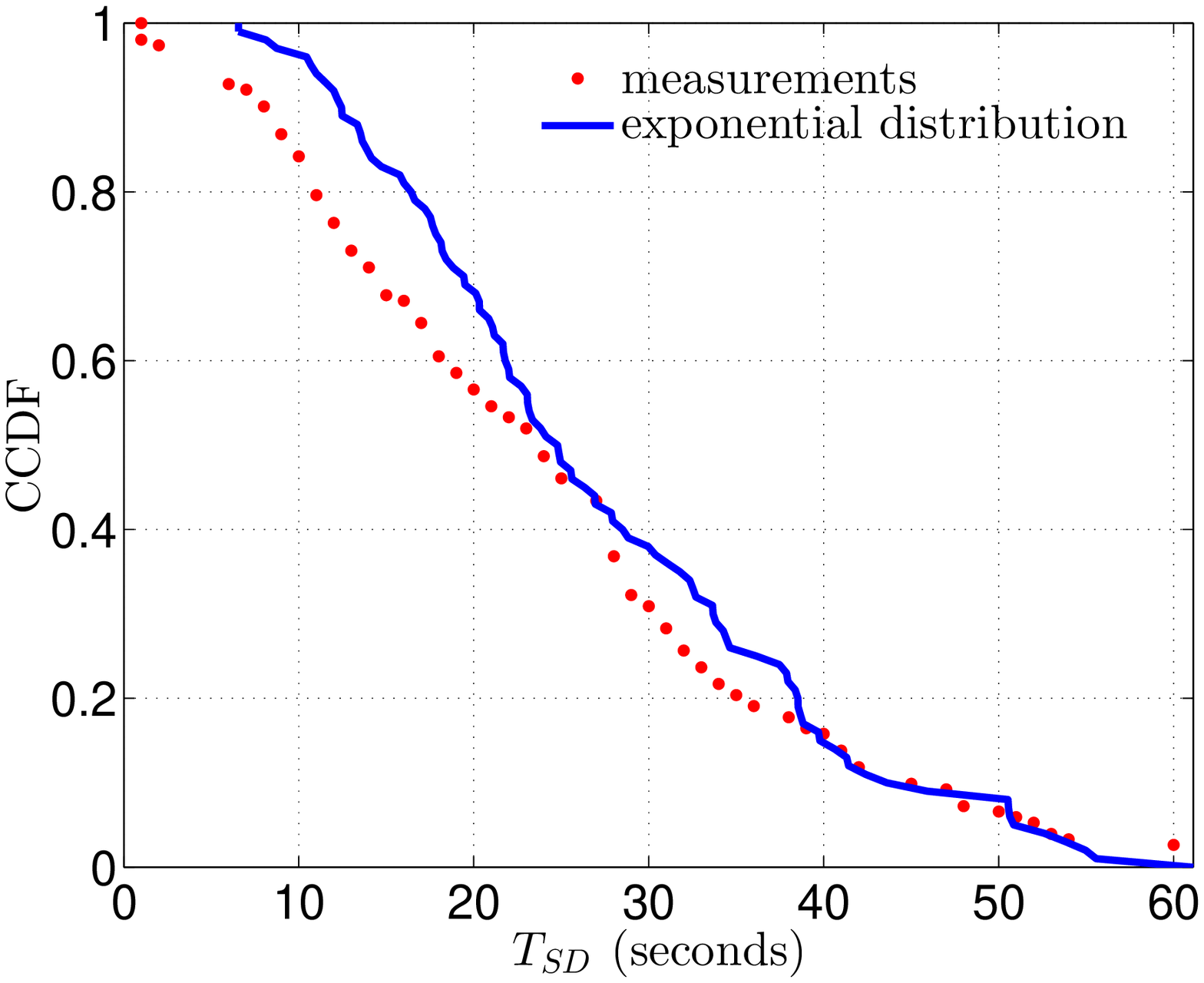}\label{fig:measurements-poisson-assumption-path-length-5}}
\caption{CCDF of the times $T_{SD}$ for SD-paths of length (a) $d=2$ and (b) $d=5$. Comparison of times $T_{SD}$ from \textit{measurements} in the Internet (where we found $E[T_{bgp}] = 6.27$), and times $T_{SD}$ generated from our model with $f_{bgp}(t)\sim$ \textit{exponential distribution}$\left(\lambda=\frac{1}{E[T_{bgp}]}\right)$.}
\label{fig:measurements-poisson-assumption}
\end{figure}

%%%%% commented text %%%%
%\begin{comment}

%%%%%%%%%%%%%%%%%%%%%%%%%%%%%%%%%%%%%%%%%%%
%%%%%%%%%%%%%%%%%%%%%%%%%%%%%%%%%%%%%%%%%%%
%%%%%%%%%%%%%%%%%%%%%%%%%%%%%%%%%%%%%%%%%%%
\section{Proof of Theorem~\ref{thm:sd-path-d-k}}\label{sec:proof-of-thm-sd-path-d-k}

\begin{proof}
Let us assume a SD-path of length $d$ and denote the ASes/nodes in the path as $n_{0}, n_{1}, ..., n_{d}$, where $n_{0}\equiv S$ and $n_{d}\equiv D$. The total number of ASes on the SD-path is $d+1$  (including nodes S and D). Let us denote as $k^{'}$, $0\leq k^{'} \leq d+1$ the number of these nodes that belong also to the SDN cluster. 

If none of the nodes comprising the SD-path belong to the SDN cluster (i.e., $k^{'} = 0$), the BGP updates propagate from $n_{0}\equiv S$ to $n_{1}$, then from $n_{1}$ to $n_{2}$, etc., till they reach the destination node $n_{d}\equiv D$. Therefore, the time $T_{SD}$ is equal to
\begin{equation}\label{eq:Tsd-sum-d-Ti_i+1}
T_{SD} = T_{n_{0},n_{1}} + T_{n_{1},n_{2}} + ... + T_{n_{d-1},n_{d}} = \sum_{i = 0}^{d-1}T_{n_{i},n_{i+1}}
\end{equation}
and since the times $T_{n_{i},n_{i+1}}$ are iid random variables, i.e., $T_{n_{i},n_{i+1}}\sim f_{bgp}(t)$, (Assumption~\ref{assumption:t-bgp}), the expectation of $T_{SD}$ is
\begin{equation}\label{eq:Tsd-sum-d-Tbgp}
E[T_{SD}|d,k^{'}=0] = \sum_{i = 0}^{d-1}E\left[T_{n_{i},n_{i+1}}\right] = d\cdot E[T_{bgp}]
\end{equation}

Now assume that the node $n_{j}$, $j=1,...,d$, is the only node on the SD-path that belongs in the SDN cluster (i.e., $k^{'}=1$). Let $T_{1} = \sum_{i=1}^{j-1}T_{n_{i},n_{i+1}}$ be the time needed for the update to propagate from $n_{0}\equiv S$ to $n_{j}$, and $T_{2} = \sum_{i=j}^{d-1}T_{n_{i},n_{i+1}}$ the time needed for the update to propagate from $n_{j}$ to the destination $n_{d} \equiv D$.

%The BGP update will start propagating from $n_{0}\equiv S$ to $n_{1}$, then from $n_{1}$ to $n_{2}$, etc., and it reaches at time $T_{1} = \sum_{i=0}^{j-1}T_{n_{i},n_{i+1}}$ at node $n_{j}$. When $n_{j}$ is informed about the BGP update, it forwards the update to the next node in the path $n_{j+1}$; the update propagates further in the path till it reaches the destination node. Let $T_{2} = \sum_{i=j}^{d-1}T_{n_{i},n_{i+1}}$ be the time needed for the update to propagate from $n_{j}$ to the destination.

The node $n_{j}$ is first informed about the BGP update at time $T^{'}\leq T_{1}$: either from the previous node in the path ($T^{'}=T_{1}$), or at an earlier time ($T^{'}<T_{1}$) from the SDN cluster, if the SDN cluster has received (through another path) the BGP update earlier.
%\footnote{The BGP updates propagate only towards one direction on a path (from node S to D). Hence, even if $n_{j}$ receives the update (from the SDN controller) earlier than node $n_{j-1}$, the update is not forwarded to (or accepted from) $n_{j-1}$ since this would create a loop or lead to a non-shortest path.}. 

Therefore, the total time needed for all the nodes in the SD-path to receive the BGP update can be expressed as
\begin{equation}\label{eq:Tsd-definition}
T_{SD} = max\{T_{1} , T^{'}+T_{2}\}
\end{equation}

\vspace{\baselineskip}
\noindent\underline{Lower Bound:}

To derive the lower bound of the expectation of $T_{SD}$, we take the expectations on \eq{eq:Tsd-definition} and proceed as follows.
\begin{align}
E[T_{SD}|d,k^{'}=1] 
&= E\left[max\{T_{1} , T^{'}+T_{2}\}\right] \\
&\geq E\left[max\{T_{1} , T_{2}\}\right] \label{eq:Tprime-0}\\
&\geq max\left\{E[T_{1}] , E[T_{2}]\right\} \label{eq:Expectation-vs-Max}\\
&= max\left\{E[T_{bgp}]\cdot d_{1} , E[T_{bgp}]\cdot d_{2}\right\} \label{eq:Expectations-times-ETbgp}\\
&= E[T_{bgp}]\cdot max\left\{ d_{1} ,  d_{2}\right\} \\
&\geq E[T_{bgp}]\cdot \displaystyle\min_{d_{1},d_{2}}\left\{max\left\{ d_{1} ,  d_{2}\right\}\right\} \label{eq:min-d1-d2}
\end{align}
%\\& = E[T_{bgp}]\cdot \frac{d}{2} \label{eq:d1-d2-dover2}
which gives
\begin{equation}\label{eq:d1-d2-dover2}
E[T_{SD}|d,k^{'}=1] ~~\geq~~ \left\{
\begin{tabular}{lc}
$\displaystyle 0$				&, $d=1$\\
$\displaystyle E[T_{bgp}]\cdot \frac{d}{2}$	&, $d>1$
\end{tabular}
\right.
\end{equation}
where 
\begin{itemize}
\item \eq{eq:Tprime-0} follows since $T^{'}\geq 0$ ($T^{'}=0$ denotes the event that the SDN cluster receives the BGP update immediately after the routing change takes place).

\item The inequality of \eq{eq:Expectation-vs-Max} follows since the times $T_{1}$ and $T_{2}$ are independent random variables, and thus it holds
\begin{align*}
P\{max\{T_{1},T_{2}\}\leq t\} 	&= P\{T_{1}\leq t\}\cdot P\{T_{2}\leq t\}~~~\Rightarrow\\
P\{max\{T_{1},T_{2}\}\leq t\} 	&\leq P\{T_{i}\leq t\}~~~\Rightarrow\\
%1- P\{max\{T_{1},T_{2}\}\leq t\}	&\geq 1-P\{T_{i}\leq t\}~~~\Rightarrow\\
P\{max\{T_{1},T_{2}\}> t\} 		&\geq P\{T_{i}> t\},~~~~\forall i=\{1,2\}
\end{align*}
and for a positive r.v. $X$ it also holds that $E[X]=\int_{0}^{\infty}P\{X>x\}dx$, and thus taking the integral in the above inequality it follows
\begin{align*}
\int_{0}^{\infty}P\{max\{T_{1},T_{2}\}> t\}dt 	&\geq \int_{0}^{\infty}P\{T_{i}> t\}dt~\Rightarrow\\
E[max\{T_{1},T_{2}\}] 						&\geq E[T_{i}],~~~~\forall i=\{1,2\}
\end{align*}
or, equivalently, $E[max\{T_{1},T_{2}\}] \geq max\left\{E[T_{i}]\right\}$.

\item The expectations $E[T_{i}], ~i=\{1,2\}$ are substituted in \eq{eq:Expectations-times-ETbgp} with $E[T_{bgp}]\cdot d_{i}$ since $T_{i}$ is the sum of $d_{i}$ iid r.v. with expected value $E[T_{bgp}]$.

\item In \eq{eq:min-d1-d2} we consider all the possible combinations of $d_{1}$ and $d_{2}$ (under the condition $d_{1}+d_{2}=d$), whose max value is minimized when $d_{1}=d_{2}=\frac{d}{2}$ (\eq{eq:d1-d2-dover2}).
\end{itemize}

Now, if there are $k^{'}$ nodes in the SD-path that belong to the SDN cluster, proceeding similarly to the above case $k^{'}=1$ leads to the following generic inequality
\begin{equation}
E[T_{SD}|d,k^{'}] ~~\geq~~ \left\{
\begin{tabular}{lc}
$ 0$				&, $d\leq k^{'}$\\
$ E[T_{bgp}]\cdot \frac{d}{k^{'}+1}$	&, $d>k^{'}$
\end{tabular}
\right.
\end{equation}
which gives the lower bound of Theorem~\ref{thm:sd-path-d-k}.
%\begin{align}
%E[T_{SD}|d,k^{'}] = E[T_{bgp}]\cdot \frac{d}{k^{'}+1} \label{eq:d1-d2-dover2}
%\end{align}

\vspace{\baselineskip}
\noindent\underline{Upper Bound:}

For $k^{'}=0$, the expectation of $T_{SD}$ is given by \eq{eq:Tsd-sum-d-Tbgp}. For $k^{'}=1$, since $T^{'}\leq T_{1}$, we can use \eq{eq:Tsd-definition} and write
\begin{align}
E[T_{SD}|d,k^{'}=1] 
&= E\left[max\{T_{1} , T^{'}+T_{2}\}\right] \\
&\leq E\left[max\{T_{1} , T_{1}+T_{2}\}\right] \\
&= d\cdot E[T_{bgp}] \label{eq:upper-bound-k=1}
\end{align}
where the last equality follows from \eq{eq:Tsd-sum-d-Tbgp}.

In the case of $k^{'}>1$, it is probable that, after the SDN cluster is informed about the routing change, the BGP update propagates simultaneously on more than one sections on the SD-path. For example, in Fig.~\ref{fig:sd-path}, after the SDN cluster is informed ($n_{i}$ and $n_{j}$ receive the update at the same time), the BGP update will propagate \textit{simultaneously} in the sub-paths $n_{i}\rightarrow ...\rightarrow n_{j-1}$ and $n_{j}\rightarrow ...\rightarrow n_{d}$. This, accelerates the propagation process, and, thus, decreases the time $T_{SD}$. 

It is easy to see, that the smaller decrease (on average) on $T_{SD}$, will take place when the $k^{'}$ nodes that belong to the SDN cluster are located consecutively on the SD-path. Without loss of generality, let assume that the first $k^{'}$ nodes $n_{0},...,n_{k^{'}-1}$ are the nodes that belong to the SDN cluster, and denote the time $T_{SD}$ for this (worst) case as $T_{SD}^{max}$. Then, the time $T_{SD}^{max}$ is given by
\begin{align}
T_{SD}^{max} 	&= \sum_{i=0}^{k^{'}-2}T_{n_{i},n_{i+1}} + \sum_{i=k^{'}-1}^{d-1}T_{n_{i},n_{i+1}} = \sum_{i=k^{'}-1}^{d-1}T_{n_{i},n_{i+1}}
\end{align}
since $\sum_{i=0}^{k^{'}-2}T_{n_{i},n_{i+1}} = T_{sdn}\equiv 0$. The expectation of $T_{SD}^{max}$ is derived similarly to \eq{eq:Tsd-sum-d-Ti_i+1} and \eq{eq:Tsd-sum-d-Tbgp}, i.e., 
\begin{equation}\label{eq:T-sd-max}
E[T_{SD}^{max}] = E\left[\sum_{i=k^{'}-1}^{d-1}T_{n_{i},n_{i+1}}\right] = \left(d-(k^{'}-1)\right)\cdot E[T_{bgp}]
\end{equation}

Combining \eq{eq:Tsd-sum-d-Tbgp}, \eq{eq:upper-bound-k=1}, and \eq{eq:T-sd-max}, gives the upper bound of Theorem~\ref{thm:sd-path-d-k}.
\end{proof}

%%%%%%%%%%%%%%%%%%%%%%%%%%%%%%%%%%%%%%%%%%%
%%%%%%%%%%%%%%%%%%%%%%%%%%%%%%%%%%%%%%%%%%%
%%%%%%%%%%%%%%%%%%%%%%%%%%%%%%%%%%%%%%%%%%%

\section{Proof of Lemma~\ref{thm:P-sdn}}\label{sec:proof-of-thm-P-sdn}

\begin{proof}
Considering all the cases for which node initiates the routing change, the probability that the source node belongs to the SDN cluster (and thus $x=0$) is
\begin{equation}
P_{sdn}(0)\equiv P_{sdn}(x=0) = \frac{k}{N}
\end{equation}
If the source node does not belong to the SDN cluster, then at step $1$ there are $N-1$ bgp-eligible nodes, of which $k$ belong to the SDN cluster. This gives
\begin{equation}
P_{sdn}(1~|x>0) = \frac{k}{N-1}
\end{equation}
and, consequently,
\begin{align*}
P_{sdn}(1) = P_{sdn}(1~|x>0)\cdot P_{sdn}(x>0) = \textstyle \frac{k}{N-1}\cdot \left(1-\frac{k}{N}\right)
\end{align*}
%\begin{align*}
%P_{sdn}(1) &= P_{sdn}(1~|x>0)\cdot P_{sdn}(x>0) \\
%			& = P_{sdn}(x=1|x>0)\cdot \left(1-P_{sdn}(x=0)\right) \\
%			& = \textstyle \frac{k}{N-1}\cdot \left(1-\frac{k}{N}\right)
%\end{align*}
Proceeding recursively, we derive \eq{eq:P-sdn} that gives the probability $P_{sdn}(x)$.
\end{proof}

%%%%%%%%%%%%%%%%%%%%%%%%%%%%%%%%%%%%%%%%%%%
%%%%%%%%%%%%%%%%%%%%%%%%%%%%%%%%%%%%%%%%%%%
%%%%%%%%%%%%%%%%%%%%%%%%%%%%%%%%%%%%%%%%%%%
\section{Proof of Theorem~\ref{thm:MGF-Tc}}\label{sec:proof-of-thm-MGF-Tc}

\begin{proof}
The convergence time is $T_{c}$ is calculated by the sum of the transition times of the Markov Chain of Fig.~\ref{fig:mc-steps}, i.e., 
\begin{equation}
T_{c} = T_{1,2} + T_{2,3} + ... + T_{N-k,C} = \sum_{i =1}^{N-k} T_{i,i+1}
\end{equation}
where we denote $T_{N-k,N-k+1}\equiv T_{N-k,C}$. Hence, the MGF of $T_{c}$ is expressed as
\begin{align}
M_{T_{c}}(\theta) 
& = E\left[e^{\theta\cdot \sum_{i =1}^{N-k} T_{i,i+1}}\right] \\
& = E\left[\prod_{i =1}^{N-k} e^{\theta\cdot T_{i,i+1}}\right] \label{eq:mgf-expectation-of-product}\\
& = \sum_{x=0}^{N-k} E\left[\prod_{i =1}^{N-k} e^{\theta\cdot T_{i,i+1}} \Big| x\right] \cdot P_{sdn}(x) \label{eq:mgf-conditional-expectation}\\
& = \sum_{x=0}^{N-k} \prod_{i =1}^{N-k} E\left[ e^{\theta\cdot T_{i,i+1}} \Big| x\right] \cdot P_{sdn}(x) \label{eq:mgf-independency}\\
& = \sum_{x=0}^{N-k} \prod_{i =1}^{N-k} \left(1-\frac{\theta}{\lambda \cdot D(i|x)}\right)^{-1} \cdot P_{sdn}(x) \label{eq:mgf-exponential-variable}
\end{align}
where 
\begin{itemize}
\item In \eq{eq:mgf-conditional-expectation} we consider the conditional expectation, given that the SDN cluster receives the update at step~$x$.

\item  \eq{eq:mgf-independency} follows from the fact that the times $T_{i,i+1}$ are independent under a given $x$; due to Assumption~\ref{assumption:t-bgp-poisson}, they depend only on the number of infected nodes, which is determined by the step $i$ and the value of $x$. 

\item We derive \eq{eq:mgf-exponential-variable}, since $T_{i,i+1}$ is an exponential random variable with rate $\lambda_{i,i+1}^{'} = \lambda \cdot D(i|x)$, and the MGF of an exponential r.v. with rate $\mu$ is given by $\left(1-{\theta}/{\mu}\right)^{-1}$. 
\end{itemize}
\end{proof}

%%%%%%%%%%%%%%%%%%%%%%%%%%%%%%%%%%%%%%%%%%%
%%%%%%%%%%%%%%%%%%%%%%%%%%%%%%%%%%%%%%%%%%%
%%%%%%%%%%%%%%%%%%%%%%%%%%%%%%%%%%%%%%%%%%%
\section{Proof of Result~\ref{thm:Dix-poisson}}\label{sec:proof-of-thm-Dix-poisson}

\begin{proof}
To derive the MGF of $T_{c}$ we apply the methodology in the proof of Lemma~\ref{thm:MGF-Tc}; here, we highlight only the key points and differences from the full-mesh case. 
\begin{align}
&M_{T_{c}}(\theta) 
 = E\left[\prod_{i =1}^{N-k} e^{\theta\cdot T_{i,i+1}}\right] \label{eq:mgf-expectation-of-product-poisson}\\
& = \sum_{S\in\mathcal{P}}\sum_{x=0}^{N-k} E\left[\prod_{i =1}^{N-k} e^{\theta\cdot T_{i,i+1}} \Big| x,S\right] \cdot P\{x,S\} \label{eq:mgf-conditional-expectation-poisson}\\
& = \sum_{x=0}^{N-k}\sum_{S\in\mathcal{P}} E\left[\prod_{i =1}^{N-k} e^{\theta\cdot T_{i,i+1}} \Big| x,S\right] \cdot P\{S\}\cdot P_{sdn}(x) \label{eq:mgf-conditional-expectation-independent-Psdn-poisson}\\
& = \sum_{x=0}^{N-k}\sum_{S\in\mathcal{P}} \prod_{i =1}^{N-k} \left(1-\frac{\theta}{\lambda \cdot D(i|x,S)}\right)^{-1} \cdot P\{S\}\cdot P_{sdn}(x) \label{eq:mgf-conditional-expectation-product-poisson}\\
& = \sum_{x=0}^{N-k} E\left[ \prod_{i =1}^{N-k} \left(1-\frac{\theta}{\lambda \cdot D(i|x,S)}\right)^{-1} \right]\cdot P_{sdn}(x) \label{eq:mgf-expectation-over-S-poisson}\\
& \approx \sum_{x=0}^{N-k} \prod_{i =1}^{N-k} \left(1-\frac{\theta}{\lambda \cdot E_{\mathcal{P}}[D(i|x)]}\right)^{-1} \cdot P_{sdn}(x) \label{eq:mgf-conditional-expectation-product-average-D-poisson}
\end{align}
where
\begin{itemize}
\item After expressing the MGF in \eq{eq:mgf-expectation-of-product-poisson}, we apply the conditional expectation property to write \eq{eq:mgf-conditional-expectation-poisson}, where $x$ is the step that the SDN cluster received the BGP update, $S$ is the set of nodes that have the BGP update, and with $P\{x,S\}$ we denote the respective joint probability.
\item Since we assume the SDN cluster to be formed independently of the topology, it holds (for any topology) that the variables $x$ and $S$ are independent. Hence, $P\{x,S\} = P\{x\}\cdot P\{S\}$, where $P\{x\}\equiv P_{sdn}(x)$ and its value is given by Theorem~\ref{thm:P-sdn}. Also, we can reorder the summations over $x$ and $S$, which gives \eq{eq:mgf-conditional-expectation-independent-Psdn-poisson}.
\item \eq{eq:mgf-conditional-expectation-product-poisson} follows by making similar arguments as in the proof of Lemma~\ref{thm:MGF-Tc}, and can be written as \eq{eq:mgf-expectation-over-S-poisson}, where the expectation is taken over the set $S\in\mathcal{P}$. 
\item Since the expectation in \eq{eq:mgf-expectation-over-S-poisson} is difficult to compute (see above discussion), we approximate it with the \textit{Delta method}~\cite{Oehlert1992}. In the Delta method the expectation of a function (i.e., the product in \eq{eq:mgf-expectation-over-S-poisson}) of a random variable (i.e., $D(i|x,S)$) is approximated by the function of the expectation of the random variable (i.e., $E_{\mathcal{P}}[D(i|x)]$).
\end{itemize}

From \eq{eq:mgf-conditional-expectation-product-average-D-poisson}, it can be seen that the approximation of $M_{T_{c}}(\theta)$ is given by an expression as in Lemma~\ref{thm:MGF-Tc}, where $D(i|x)$ is replaced by $E_{\mathcal{P}}[D(i|x)]$. Moreover, it is easy to see that all the consequent results for the full-mesh network can be similarly modified for the Poisson graph case.

Now, we need only to calculate the expected bgp-degree $E_{\mathcal{P}}[D(i|x)]$: Let assume that we are at step $i$, and $n(i)$ nodes (see \eq{eq:n(i)}) have received the BGP updates; we denote the set of these nodes as $S_{i}$. A node $m\notin S_{i}$ is connected with a node $j\in S_{i}$ with probability $P(m,j)=p$ (by the definition of a Poisson graph). Hence, the probability that $m$ is a bgp-eligible node (i.e., is connected with \textit{any} of the nodes $j\in S_{i}$, where $|S_{i}| = n(i)$), is given by
\begin{align}
P(m,S_{i}) = 1- (1-p)^{n(i)}
\end{align}

Finally, we note that there are $N-n(i)$ nodes without the update, with each of them being a bgp-eligible node with any of the nodes $j\in S_{i}$ with (equal) probability $P(m,S_{i})$. As a result, the total number of bgp-eligible nodes (or, as defined in Def.~\ref{def:bgp-degree}, the \textit{bgp-degree} $D(i)$) is a binomially distributed random variable, whose expectation is given by 
\begin{equation}
E[D(i)] = (N-n(i))\cdot (1-(1-p)^{n(i)})
\end{equation}
\end{proof}

%%%%%%%%%%%%%%%%%%%%%%%%%%%%%%%%%%%%%%%%%%%
%%%%%%%%%%%%%%%%%%%%%%%%%%%%%%%%%%%%%%%%%%%
%%%%%%%%%%%%%%%%%%%%%%%%%%%%%%%%%%%%%%%%%%%
\section{Internet Topology and Routing Policies}\label{sec:simulations-internet}
To approximate the routing system of the Internet, we use a methodology similar to many previous works~\cite{Let-the-market-BGP-sigcomm-2011,how-secure-goldberg-ComNet-2014,Jumpstarting-BGP-sigcomm-2016,RPKI-deployment-2016}. We first build the Internet topology graph from a large experimentally collected dataset~\cite{AS-relationships-dataset}, and infer routing policies over existing links based on the Gao-Rexford conditions~\cite{stable-internet-routing-TON-2001}.

\subsection{Building the Internet Topology}
We build the Internet topology graph from the AS-relationship dataset of CAIDA~\cite{AS-relationships-dataset}, which is collected based on the methodology of~\cite{AS-relationships-IMC-2013} and enriched with many extra peering (p2p) links~\cite{multilateral-peering-conext-2013}. The dataset contains a list of AS pairs with a peering link, which is annotated based on their relationship as \textit{c2p} (\textit{customer to provider}) or \textit{p2p} (\textit{peer to peer}).

\subsection{Selecting Routing Policies}
When an AS learns a new route for a prefix (or, announces a new prefix), it updates its routing table and, if required, sends BGP updates to its AS neighbors. The update and export processes are defined by its routing policies. Similarly to previous works~\cite{Let-the-market-BGP-sigcomm-2011,how-secure-goldberg-ComNet-2014,Jumpstarting-BGP-sigcomm-2016,RPKI-deployment-2016}, we select the routing policies based on the Gao-Rexford conditions that guarantee BGP convergence and stability~\cite{stable-internet-routing-TON-2001}:
\begin{description}
\item[C.1] Paths learned from customers are preferred to paths learned from peers or providers. Paths learned from peers are preferred to paths learned from providers.
\item[C.2] Between paths that are equivalent with respect to \textbf{C.1}, shorter paths (in number of AS-hops) are preferred.
\item[C.3] Between paths that are equivalent with respect to \textbf{C.1} and \textbf{C.2}, the path learned from the AS neighbor with the highest \textit{local preference} is preferred.
\item[C.4] Paths learned from customers, are advertised to all AS neighbors. Paths learned from peers or providers, are advertised only to customers.
\end{description}

In practice, the local preferences (see, \textbf{C.3}) are selected by an AS based on factors related to its intra-domain topology, business agreements, etc. Since it is not possible to know and emulate the real policies for every AS, we assign randomly the local preferences.

%\end{comment}

\end{document}